\documentclass{iopart}
\usepackage{amssymb,amsfonts,amsthm}
\usepackage{enumerate}
\usepackage{tikz}
\theoremstyle{plain}
\newtheorem{theorem}{Theorem}[section]
\newtheorem{lemma}[theorem]{Lemma}
\newtheorem{proposition}[theorem]{Proposition}
\theoremstyle{remark}
\newtheorem{remark}[theorem]{Remark}
\newcommand{\ZZ}{\mathbb{Z}}
\newcommand{\RR}{\mathbb{R}}
\newcommand{\CC}{\mathbb{C}}
\newcommand{\NN}{\mathbb{N}}
\newcommand{\EE}{\mathbb{E}}
\newcommand{\drm}{\ensuremath{\mathrm{d}}}
\newcommand{\p}{\ensuremath{\mathbb{P}}}
\newcommand{\re}{\ensuremath{\mathrm{Re}\,}}
\newcommand{\im}{\ensuremath{\mathrm{Im}\,}}
\newcommand{\euler}{\mathrm{e}}
\newcommand{\PP}{\mathbb{P}}
\renewcommand{\i}{\ensuremath{{\mathrm{i}}}}
\newcommand{\sprod}[2]{\left\langle#1,#2 \right \rangle}
\newcommand{\abs}[1]{\left\vert#1\right\vert}
\newcommand{\supp}{\ensuremath{\mathrm{supp}\,}}
\newcommand{\dist}{\ensuremath{\mathrm{dist}\,}}
\newcommand{\Pro}{\ensuremath{P}}
\newcommand{\BIGOP}[1]{\mathop{\mathchoice%
{\raise-0.22em\hbox{\huge $#1$}}%
{\raise-0.05em\hbox{\Large $#1$}}{\hbox{\large $#1$}}{#1}}}
\newcommand{\BIGboxplus}{\mathop{\mathchoice%
{\raise-0.35em\hbox{\huge $\boxplus$}}%
{\raise-0.15em\hbox{\Large $\boxplus$}}{\hbox{\large $\boxplus$}}{\boxplus}}}
\newcommand{\bigtimes}{\BIGOP{\times}}
\begin{document}
\title[Localization for models on $\mathbb{Z}$  with single-site potentials of finite support]{Localization via fractional moments for models on $\mathbb{Z}$  with single-site potentials of finite support \\ {\small \it Dedicated to the memory of Pierre Duclos}}
\author{A. Elgart$^1$, M. Tautenhahn$^2$ and I. Veseli\'c$^2$}
\address{$^1$ Department of Mathematics, Virginia Tech., Blacksburg, VA, 24061, USA}
\address{$^2$ Fakult\"at f\"ur Mathematik, TU Chemnitz, 09126 Chemnitz, Germany}
\eads{\mailto{aelgart@vt.edu}, \mailto{martin.tautenhahn@mathematik.tu-chemnitz.de}, \mailto{ivan.veselic@mathematik.tu-chemnitz.de}}
\begin{abstract}
One of the fundamental results in the theory of localization for discrete
Schr\"o\-dinger operators with random potentials is the exponential decay of
Green's function and the absence of continuous spectrum.
In this paper we provide a new variant of these results for
one-dimensional alloy-type potentials with finitely supported
sign-changing single-site potentials using the fractional moment method.
\end{abstract}
\submitto{JPA}
%
%
%
%
%
%
\section{Introduction}
Anderson models on the lattice are discrete Schr\"odinger operators with random potentials. Such models have been studied since a long time in computational and theoretical physics, as well as in mathematics. One of the fundamental results for these
models is the physical phenomenon of \emph{localization}. There are various manifestations of localization: exponential decay of Green's function, absence of diffusion, spectral localization (meaning almost sure absence of continuous
spectrum), exponential decay of generalized eigensolutions, or non-spreading of wave packets.  Such properties have been established exclusively (apart from one-dimensional situations) by two different methods, the multiscale analysis and the fractional moment method. The \emph{multiscale analysis} (MSA) was invented by Fr\"ohlich and Spencer in \cite{FroehlichS1983}, while the \emph{fractional moment method} (FMM) was introduced by Aizenman and Molchanov \cite{AizenmanM1993}.
\par
In this paper we focus our attention on correlated Anderson models. More precisely, we develop the FMM for a one-dimensional discrete Schr\"odinger operator with random potential of alloy-type. In this model, the potential at the lattice site $x \in \mathbb Z$ is defined by a finite linear combination $V_\omega (x) = \sum_k \omega_k u(x-k)$ of independent identically distributed (i.\,i.\,d.) random coupling constants $\omega_k$ having a bounded density. The function $u(\cdot - k)$ is called single-site potential and may be interpreted as a finite  interaction range potential associated to the lattice site $k \in \mathbb Z$. Consequently, for the model under consideration the potential values at different sites are not independent random variables.  Let us stress that we have no sign assumption on the single-site potential, thus the correlations may be negative.
\par
For such models we prove in one space dimension and at all energies a so-called fractional moment bound, i.\,e. exponential off-diagonal decay of an averaged fractional power of Green's function. The restriction to the one-dimensional case allows an elegant and short proof in which the basic steps---decoupling and averaging---are particularly transparent. Currently we are working on the extension of our result to the multi-dimensional case.
\par
A second result concerns a criterion of exponential localization,  {i.\,e.} the fact that in a certain interval there is no continuous spectrum and all eigenfunctions decay exponentially almost surely. Our proof uses an idea developed first in the context of MSA, but actually establishes exponential localization directly from fractional moment
bounds on the Green's function, without going through the induction step of the MSA.
\par
Let us discuss  which localization results established previously apply to the case of correlated potentials. The papers \cite{DreifusK1991} and \cite{AizenmanM1993,AizenmanSFH2001} derive localization for  Anderson models with correlated potentials, using the MSA or the FMM, respectively. However, all of them require quite stringent conditions on the conditional distribution of the potential value at a site  conditioned on the remaining sites. These are typically not satisfied for the discrete alloy-type potential whenever the coupling constants are bounded random variables, cf.{} \cite{TautenhahnV2010}. Thus, our results are not covered by those in \cite{DreifusK1991,AizenmanM1993,AizenmanSFH2001}.
\par
For continuous alloy-type models spectral localization can be derived via MSA, as soon as one has verified a Wegner-type bound and an initial scale estimate. There are certain energy/disorder regimes where this has been achieved for  sign-changing single-site potentials. All of them require that the random variables have a bounded density. For weak disorder \cite{Klopp2002} establishes localization near non-degenerate simple band edges. For the bottom of the spectrum localization holds as well. In the case that the bottom of the spectrum is $-\infty$ this has been proven in \cite{Klopp1995}, while for lower bounded operators it follows from a combination of \cite{Klopp1995,HislopK2002} and \cite{KloppN2009}. 
See also \cite{Veselic2002,KostrykinV2006,Veselic2010b} for related results, which however apply only to particular single-site potentials of a generalized step-function form. 
It should be emphasized that all these proof use monotonicity at some stage of the argument which is not the 
case for our proof of exponential decay of fractional moments. This is discussed in more detail below. In \cite{BuschmanS2001,Stolz2002} localization is established  for
alloy-type models on $\RR$ with sign-changing single-site potentials using genuinly one-dimensional techniques. Thus these results are related to ours. However, to our knowledge, fractional moment bounds have not been established for  alloy-type models  with sign-changing single-site potential so far (neither in the continuous nor the
discrete setting). Let us stress that even for models where the MSA is well established, the implementation of the FMM gives new insights and slightly stronger
results. For instance, the paper \cite{AizenmanENSS2006} concerns models for which the MSA was developed much earlier.
\par
We would like to bring one particular feature of our proof of exponential decay of fractional moments to the attention of the reader. Contrary to standard proofs of localization either via multiscale analysis or via the fractional moment method this  proof nowhere uses a monotonous spectral averaging, a monotonous Wegner estimate, or any other kind of monotonicity argument. The key tool which allows us to disregard monotonicity issues is an averaging result for determinants, formulated in Lemma \ref{lemma:det}.
\par
Let us elaborate on the role played by monotonicity in previous arguments in some more detail. The first approach to localization for alloy-type models which does not use monotonicity is the one pursued in \cite{Stolz2002,BuschmanS2001}. However, this method does not yield such a strong bound on the Green's function as Theorem \ref{thm:result1} below. Furthermore, the method \cite{BuschmanS2001,Stolz2002} is based on the Pr\"ufer coordinate, a quantity which is only defined for one-dimensional models. The fractional moment method is a tool suited for models in arbitrary dimension, albeit so far we have implemented in only in one dimension for our model. 
\par
A result related to bounds on the fractional moments of the Green's function is a Wegner estimate. Such estimates have been developed for alloy-type models with sign-changing single-site potentials. There are two methods at disposal to derive a Wegner bound in this situation, one developed in \cite{Klopp1995,HislopK2002,Klopp2002} and the other in \cite{Veselic2002,KostrykinV2006,Veselic2010b}. Now both of them use monotonicity at some stage of the argument. Let us first discuss aspects of the method in \cite{Klopp1995, HislopK2002,Klopp2002}, restricting ourselves for simplicity to the energy region near the bottom of the spectrum. (\cite{HislopK2002,Klopp2002} have also results about internal spectral edges in the weak disorder regime.) There it is shown that negative eigenvalues of certain auxiliary operators have a negative derivative with respect to an appropriately chosen vector field. Thus a  monotonicity property is  established for appropriately chosen spectral subspaces. On the other hand, the idea of \cite{Veselic2002,KostrykinV2006} consist in finding a linear combination of single-site potentials which is non-negative and averaging w.\,r.\,t.~the associated
transformed coupling constants. This is obviously again a monotonicity argument.
\par
Recently Bourgain \cite{Bourgain2009} has developed a method to establish Wegner estimates for Anderson models with matrix-valued potentials. These models, like ours, lack monotonicity. To overcome this difficulty Bourgain uses analyticity and subharmonicity properties of the relevant matrix-functions. A key ingredient is (a multidimensional version of) Cartan's Theorem. If one likes, one can interpret our Lemma
\ref{lemma:det} as a very specific and explicite version of Cartan's Theorem. 
\par
Using methods from dynamical systems Sadel and Schulz-Baldes prove in \cite{SadelSB2008} positivity of Lyapunov exponents and at most logarithmic growth of quantum dynamics for one-dimensional random potentials with correlations. 
However, due to the assumptions on the underlying probability space it seems that our model can not be reduced to the ones in \cite{SadelSB2008}. 
Avila and Damanik have related results as announced and sketched in \cite{Damanik2007, Damanik-Oberwolfach}.
\par
Finally, another recent development concerns localization results for discrete alloy-type models based on the MSA.
After the first version of the present paper was completed, Veseli\'c \cite{Veselic2010b}  
obtained a Wegner estimate for the multi-dimensional 
analog of the model considered in this paper, which can be used to obtain exponential spectral localization at large disorder.  
For this results the distribution of the coupling constants has to satisfy an analog of condition (a) in Theorem 2.2.{} below.
In a preprint, of which we learned after the final version of our paper was submitted, 
Kr\"uger \cite{Krueger} obtained results about spectral and dynamical localization for a class of multi-dimensional models which includes ours as a subclass.
This work uses some of the ideas of \cite{Bourgain2009} mentioned above and establishes Wegner-like estimates 
without the use of monotonicity.

%
%
%
%
%
%
\section{Model and results} \label{sec:model}
We consider a one-dimensional Anderson model. This is the random discrete Schr\"o{}\-ding\-er operator
\begin{equation} \label{eq:hamiltonian}
 H_\omega := -\Delta + V_\omega , \quad \omega \in \Omega,
\end{equation}
acting on $\ell^2 (\ZZ)$, the space of all square-summable sequences indexed by $\ZZ$ with an inner product $\sprod{\cdot}{\cdot}$. Here, $\Delta: \ell^2 \left(\ZZ\right) \to \ell^2 \left(\ZZ\right)$ denotes the discrete Laplace operator and $V_\omega : \ell^2 \left(\ZZ\right) \to \ell^2 \left(\ZZ\right)$ is a random multiplication operator. They are defined by
\begin{equation*}
\left(\Delta \psi \right) (x) := \sum_{\abs{e} = 1} \psi (x+e) \quad \mbox{and} \quad
\left( V_\omega \psi \right) (x) :=  V_\omega (x) \psi (x)
\end{equation*}
and represent the kinetic energy and the random potential energy, respectively.
We assume that the probability space has a product structure $\Omega :=
\bigtimes_{k \in \ZZ} \RR$ and is equipped with the probability measure $\drm
\p (\omega) := \prod_{k \in \ZZ} \rho (\omega_k) \drm \omega_k$ where
$\rho \in L^\infty (\RR) \cap L^1 (\RR)$ with $\| \rho \|_{L^1}$ = 1.
Hence, each element $\omega$ of $\Omega$ may be represented as a collection
$\{\omega_k\}_{k \in \ZZ}$ of {i.\,i.\,d.}
random variables, each distributed with the density $\rho$. The symbol
$\mathbb{E}\{\cdot\}$ denotes the expectation with respect to the probability
measure, i.\,e. $\mathbb{E} \{\cdot\} := \int_\Omega (\cdot) \drm
\p (\omega)$. For a set $\Gamma \subset \ZZ$, $\mathbb{E}_\Gamma\{\cdot \}$
denotes the expectation with respect to $\omega_k$, $k \in \Gamma$. That is,
$\mathbb{E}_{\Gamma} \{\cdot\} := \int_{\Omega_\Gamma} (\cdot)
\prod_{k \in \Gamma} \rho(\omega_k) \drm \omega_k$ where $\Omega_\Gamma
:= \bigtimes_{k \in \Gamma} \RR$. Let the \emph{single-site potential}
$u : \ZZ \to \RR$ be a function with finite and non-empty support $\Theta
:= \supp u = \{k \in \ZZ : u(k) \not = 0 \}$. We assume that the random
potential $V_\omega $ has an alloy-type structure, {i.\,e.}
\begin{equation*}
V_\omega (x) := \sum_{k \in \ZZ} \omega_k u (x-k)
\end{equation*}
at a lattice site $x \in \ZZ$ is a linear combination of the i.\,i.\,d. random
variables $\omega_k$, $k\in\ZZ$, with coefficients provided by the single-site
potential. For this reason we call  the Hamiltonian \eref{eq:hamiltonian} sometimes
a discrete alloy-type model. The function $u(\cdot - k)$ may be interpreted as a finite range potential associated to the lattice site $k\in\ZZ$. The Hamiltonian
\eref{eq:hamiltonian} is possibly unbounded, but self-adjoint on a dense
subspace of $\ell^2 (\ZZ)$, see e.\,g. \cite{Kirsch-08}. Finally, for the
operator $H_\omega$ in \eref{eq:hamiltonian} and $z \in \CC \setminus \sigma
(H_\omega)$ we define the corresponding \emph{resolvent} by $G_\omega (z)
:= (H_\omega - z)^{-1}$. For the \emph{Green's function}, which assigns
to each  $(x,y) \in \ZZ \times \ZZ$ the corresponding matrix element of the
resolvent, we use the notation
\begin{equation} \label{eq:greens}
G_\omega (z;x,y) := \sprod{\delta_x}{(H_\omega - z)^{-1}\delta_y}.
\end{equation}
For $\Gamma \subset \ZZ$, $\delta_k \in \ell^2 (\Gamma)$ denotes the
Dirac function given by $\delta_k (k) = 1$ for $k \in \Gamma$ and
$\delta_k (j) = 0$ for $j \in \Gamma \setminus \{k\}$. The
quantities $\| \rho \|_\infty^{-1}$ and (in the case that
$\rho$ is weakly differentiable) $\| \rho' \|_{L^1}^{-1}$ may
be understood as a measure of the disorder present in the model (a
small value of norms corresponds to the strong disorder). Our
results in the case of strong disorder are the following three
theorems.
\begin{theorem} \label{thm:result1}
 Let $n \in \NN$, $\Theta = \{0,\dots,n-1\}$, $s \in (0,1)$, and
 $\| \rho \|_\infty$ be sufficiently small.
Then there exist constants $C,m \in (0,\infty)$ such that for all $x,y \in \ZZ$ with $|x-y| \geq n$ and all $z \in \CC \setminus \RR$,
\begin{equation} \label{eq:result1}
 \mathbb{E} \bigl\{ | G_\omega (z;x,y) |^{s/n} \bigr\} \leq C{\rm e}^{-m \abs{x-y}} .
\end{equation}
\end{theorem}
\begin{theorem} \label{thm:result2}
 Let $n \in \NN$, $\Theta\subset\ZZ$ finite with $\min \Theta = 0$ and $\max \Theta = n-1$, $r$
as in Eq.{} \eref{eq:r} (the width of the largest gap in $\Theta$),
and $s \in (0,n/(n+r))$. Assume
\begin{enumerate}[(a)]
 \item $\rho \in W^{1,1} (\RR)$ with $\| \rho' \|_{L^1}$ sufficiently small, \textbf{or}
 \item $\supp \rho$ compact with $\| \rho \|_\infty$ sufficiently small.
\end{enumerate}
Then there exist constants $C,m \in (0,\infty)$ such that the bound \eref{eq:result1} holds true for all $x,y \in \ZZ$ with $| x-y | \geq 2(n+r)$ and all $z \in \CC \setminus \RR$.
\end{theorem}
\begin{theorem} \label{theorem:loc}
 Let $\supp \rho$ be compact with $\| \rho \|_\infty$ sufficiently small. Then $H_\omega$ has almost surely only pure point spectrum with exponentially decaying eigenfunctions.
\end{theorem}
The difference between Theorem~\ref{thm:result1} and Theorem~\ref{thm:result2} is the following: In Theorem~\ref{thm:result1} we assume that $\Theta$ is finite and \emph{connected} (cf. Section 3). The latter condition can be dropped if $\rho$ is sufficiently regular, cf.~Theorem~\ref{thm:result2}. A quantitative version of Theorem~\ref{thm:result1} is proven in Section~\ref{sec:finitness} and \ref{sec:exp}, compare also Theorem~\ref{theorem:exp}. A quantitative version of Theorem~\ref{thm:result2} is stated and proven in Section~\ref{sec:gen}.
\par
We can actually apply Theorems \ref{thm:result1} and \ref{thm:result2} to arbitrary $\Theta$ with $\max \Theta-\min \Theta= n-1$. In this situation a translation of the indices of the random variables $\{\omega_k\}_{k \in \ZZ}$ by $\min\Theta$ transforms  the model to the case $\min \Theta = 0$ and $\max \Theta = n-1$. Note that $\min\Theta $ and $\max\Theta $ are well defined since $\Theta \subset \RR$ is finite.
\begin{remark} \label{remark:localization}
\begin{enumerate}[(i)]
\item
Our proof gives estimates about fractional moments of \emph{certain} matrix elements
of the resolvent for somewhat more general models. Let us formulate this class of random potentials next. Assume that $V_\omega:= V_\omega^{(1)} +V_\omega^{(2)}$ where
$V_\omega^{(1)}, V_\omega^{(2)}\colon \ZZ\to \RR$ are  potentials indexed by the random parameter $\omega$ in some probability space $ \Omega$. Assume that $u \colon \ZZ \to \RR$ has support equal to $\{0,\dots , n-1\}$, and that there exists a sequence $\lambda_k\colon \Omega \to \RR$ of i.\,i.\,d. random variables indexed by $k \in n \ZZ$, each being distributed according to a density  $\rho \in L^\infty(\RR)$.
Assume that $V_\omega^{(1)}(x) = \sum_{ k \in  n\ZZ} \lambda_k (\omega ) u (x-k)$ and that $V_\omega^{(2)}$ is uniformly bounded on $\Omega \times \ZZ$, but otherwise arbitrary. If $F \colon \Omega \to [0,\infty)$ is a random variable  we denote its
 average over all random variables $\lambda_k$, $k \in n \ZZ$, by $\EE^{(1)}(F) := \int F(\omega) \prod_{k \in n\ZZ} \rho(\omega_k) \drm \omega_k$, where the domain  of integration is $\bigtimes_{k \in n\ZZ} \RR$. It follows directly from the iterative application of Lemma \ref{lemma:finitness1} that for all $ p \in \NN$ and for the constant $C_{u,\rho}$ defined in \eref{eq:cucrho} we have
\begin{equation}\label{e:subsequence}
 \EE^{(1)} \bigl\{| G_\omega (z;0,np-1) |^{s/n}\bigr\} \leq C_{u,\rho}^{p}.
\end{equation}
A  decomposition of the type $V_\omega:= V_\omega^{(1)} +V_\omega^{(2)}$ is implicitly used in the proof of Theorem \ref{thm:result2}, given in Section~\ref{sec:gen}.
Note, that in this particular situation the two stochastic processes $V_\omega^{(1)}, V_\omega^{(2)}$ are \emph{not} independent from each other. If $V_\omega^{(2)}\equiv 0$ then the full potential $V_\omega$ equals $\sum_{ k \in  n\ZZ} \lambda_k (\omega ) u (x-k)$. Hence, in this case the bound \eref{e:subsequence} also holds true. 
\item The statements of Theorems  \ref{thm:result1} and \ref{thm:result2} concern only off-diagonal elements. If we assume that $\rho$ has compact support, $\mathbb{E} \bigl\{ | G_\omega (E+i0;x,y) |^{s} \bigr\}$ is finite for any $x,y \in \ZZ$ and $s>0$ sufficiently small. This is proven in Section~\ref{sec:apriori}.
\item Thus in this situation we have full control over fractional moments, which for the
usual Anderson model with {i.\,i.\,d.} potential values suffices to prove spectral and dynamical localization. However, for our model neither dynamical nor spectral localization can be directly inferred using the existent methods in \cite{SimonW1986,Aizenman1994,Graf1994,AizenmanENSS2006}, see Section~\ref{sec:loc}.
The reason is that the random variables $V_\bullet (x)$, $x \in \ZZ$, are not independent, while the dependence of $H_\omega$ on the {i.\,i.\,d.} variables $\omega_x$, $x \in \ZZ$, is not monotone.
\item In Section~\ref{sec:loc} we provide a new criterion for spectral localization without applying the multiscale analysis. It deduces from fractional moment bounds and the fact that the set of generalized eigenvalues has full spectral measure almost sure exponential decay of eigenfunctions. In fact it can be extended to more general random potentials, as long as the correlation length is finite.
\item In this context it is natural to ask whether it is possible to extract a positive part from the random potential in such a way, that the original methods for deriving fractional moment bounds apply. It turns out that this is not possible in general (even in one space dimension), but that the corresponding class of single-site potentials can be characterized in the following way:
If the polynomial $p_u(x): = \sum_{k=0}^{n-1} u(k) \, x^k$ does not vanish on $[0, \infty)$ it is possible to extract from $V_\omega$ a positive single-site potential with certain additional properties. In this situation  the method of \cite{AizenmanENSS2006} applies and gives exponential decay of the fractional moments of the Green's function. This is worked out in detail in \ref{sec:monotone}.
\end{enumerate}
\end{remark}
%
%
%
%
%
%
\section{Fractional moment bounds for Green's function} \label{sec:finitness}
In this section we present fractional moment bounds for Green's function. A very useful observation is that ``important'' matrix elements of the resolvent are given by the inverse of a determinant. The latter can be controlled using the following spectral averaging lemma for determinants.
\begin{lemma} \label{lemma:det}
 Let $n \in \NN$ and $A, V \in \CC^{n \times n}$ be two matrices and assume that $V$ is invertible. Let further $0 \leq \rho \in L^1(\RR) \cap L^\infty (\RR)$ and $s \in (0,1)$. Then we have for all $\lambda > 0$ the bound
\begin{eqnarray}
 \int_{\RR}  \abs{\det (A + rV)}^{-s/n} & \rho (r) \drm r
\leq \abs{\det V}^{-s/n} \| \rho \|_{L^1}^{1-s} \|\rho\|_{\infty}^{s} \frac{2^{s} s^{-s}}{1-s} \label{eq:det1} \\[1ex]
&\leq \abs{\det V}^{-s/n}\Bigl( \lambda^{-s} \|\rho\|_{L^1} + \frac{2 \lambda^{1-s}}{1-s} \|\rho\|_\infty  \Bigr) \label{eq:det2} .
\end{eqnarray}
\end{lemma}
\begin{proof}
 Since $V$ is invertible, the function $r \mapsto \det (A + rV)$ is a polynomial of order $n$ and thus the set $\{r \in \mathbb{R} \colon A + rV \mbox{ is singular}\}$ is a discrete subset of $\mathbb{R}$ with Lebesgue measure zero. We denote the roots of the polynomial by $z_1,\dots , z_n \in \CC$. By multilinearity of the determinant we have
\[
 \abs{\det (A + rV)} = \abs{\det V} \prod_{j=1}^n | r - z_j | \geq
 \abs{\det V} \prod_{j=1}^n | r - \re{z_j}| .
\]
The H\"older inequality implies for $s \in (0,1)$ that
\begin{equation*} \fl
 \int_\RR \abs{\det (A + rV)}^{-s/n} \rho (r) \drm r \leq \abs{\det V}^{-s/n} \prod_{j = 1}^n \left( \int_\RR | r - \re z_j|^{-s} \rho (r) \drm r  \right)^{1/n} .
\end{equation*}
For arbitrary $\lambda > 0$ and all $z \in \RR$ we have
\begin{eqnarray*}
\int_\RR \frac{1}{\abs{r - z}^{s}} \rho(r)\drm r &=  \int\limits_{\abs{r - z} \geq \lambda} \frac{1}{\abs{r - z}^{s}} \rho(r)\drm r + \int\limits_{\abs{r - z} \leq \lambda} \frac{1}{\abs{r - z}^{s}} \rho(r)\drm r \\[1ex]
& \leq \lambda^{-s} \|\rho\|_{L^1} + \|\rho\|_\infty \frac{2 \lambda^{1-s}}{1-s}
\end{eqnarray*}
which gives Ineq.{} \eref{eq:det2}. We now choose $\lambda = s \| \rho \|_{L^1} / (2 \| \rho \|_\infty)$ (which minimises the right hand side of Ineq.{} \eref{eq:det2}) and obtain Ineq.{} \eref{eq:det1}.
\end{proof}
In order to use the estimate of Lemma \ref{lemma:det} for our infinite-dimensional operator $G_\omega (z)$, we will use a special case of the Schur complement formula (also known as Feshbach formula or Grushin problem), see e.\,g. \cite[appendix]{BellissardHS2007}. Before providing such a formula, we will introduce some more notation.
Let $\Gamma_1 \subset \Gamma_2 \subset \ZZ$. We define the operator $\Pro_{\Gamma_1}^{\Gamma_2} : \ell^2 (\Gamma_2) \to \ell^2 (\Gamma_1)$ by
\[
 \Pro_{\Gamma_1}^{\Gamma_2} \psi := \sum_{k \in \Gamma_1} \psi (k) \delta_k .
\]
Note that the adjoint $(\Pro_{\Gamma_1}^{\Gamma_2})^* : \ell^2 (\Gamma_1) \to \ell^2 (\Gamma_2)$ is given by $(\Pro_{\Gamma_1}^{\Gamma_2})^* \phi = \sum_{k \in \Gamma_1} \phi (k) \delta_k$. If $\Gamma_2 = \ZZ$ we will drop the upper index and write $\Pro_{\Gamma_1}$ instead of $\Pro_{\Gamma_1}^{\ZZ}$.
For an arbitrary set $\Gamma \subset \ZZ$ we define the restricted operators $\Delta_\Gamma, V_\Gamma, H_\Gamma:\ell^2 (\Gamma) \to \ell^2 (\Gamma)$ by
\[
 \Delta_\Gamma := \Pro_\Gamma \Delta \Pro_\Gamma^\ast \quad \mbox{and} \quad V_\Gamma := \Pro_\Gamma V_\omega \Pro_\Gamma^\ast .
\]
Furthermore, we define $G_\Gamma (z) := (H_\Gamma - z)^{-1}$ and $G_\Gamma (z;x,y) := \bigl\langle \delta_x, G_\Gamma (z) \delta_y \bigr\rangle$ for $z \in \CC \setminus \sigma (H_\Gamma)$ and $x,y \in \Gamma$.
For an operator $T:\ell^2 (\Gamma) \to \ell^2 (\Gamma)$ the symbol $[T]$ denotes the matrix representation of $T$ with respect to the basis $\{\delta_k\}_{k \in \Gamma}$.
By $\partial \Gamma$ we denote the interior vertex boundary of the set $\Gamma$, i.\,e. $\partial \Gamma := \{k \in \Gamma : \#  \{ j \in \Gamma : \abs{j-k} = 1\} < 2 \}$.
For finite sets $\Gamma \subset \ZZ$, $\abs{\Gamma}$ denotes the number of elements of $\Gamma$. A set $\Gamma \subset \ZZ$ is called \textit{connected}
if $\partial \Gamma\subset  \{\inf\Gamma, \sup\Gamma\}$. In particular, $\ZZ$ is a connected set.
\begin{lemma} \label{lemma:fraction}
Let $\Gamma \subset \ZZ$ and $\Lambda \subset \Gamma$ be finite and connected. Then we have the identity
\begin{equation*}
 G_\Gamma (z;x,y) = \bigl \langle \delta_x , (H_{\Lambda} - B_\Gamma^\Lambda - z )^{-1} \delta_y \bigr \rangle
\end{equation*}
for all $z \in \CC \setminus \sigma(H_\Gamma) $ and all $x,y \in \Lambda$, where $B_\Gamma^\Lambda:\ell^2 (\Lambda) \to \ell^2 (\Lambda)$ is specified in Eq.{} \eref{eq:bij}. Moreover, the operator $B_\Gamma^\Lambda$ is diagonal and does not depend on $V_\omega (k)$, $k \in \Lambda$.
\end{lemma}
An analogous statement for arbitrary dimension was established in \cite{ElgartG}.
\begin{proof}
Since $\Lambda$ is finite, $H_\Lambda$ is bounded and the Schur complement formula gives
\begin{equation*} 
 \Pro_{\Lambda}^\Gamma (H_\Gamma - z)^{-1} \bigl(\Pro_{\Lambda}^\Gamma\bigr)^\ast = \bigl( H_{\Lambda} - z - B_\Gamma^\Lambda \bigr)^{-1} ,
\end{equation*}
where
\[
 B_\Gamma^\Lambda = \Pro_{\Lambda}^\Gamma \Delta_\Gamma \bigl(\Pro_{\Gamma \setminus \Lambda}^\Gamma\bigr)^\ast (H_{\Gamma \setminus \Lambda} - z)^{-1} \Pro_{\Gamma \setminus \Lambda}^\Gamma \Delta_\Gamma \bigl(\Pro_{\Lambda}^\Gamma\bigr)^\ast .
\]
It is straightforward to calculate that the matrix elements of $B_\Gamma^\Lambda$ are given by
\begin{equation} \label{eq:bij}
\sprod{\delta_x}{B_\Gamma^\Lambda \delta_y} =
\cases{
\sum_{k \in \Gamma \setminus \Lambda : \atop \abs{k-x} = 1} \sprod{\delta_k}{(H_{\Gamma \setminus \Lambda} - z)^{-1}\delta_k} & if $x=y \in \partial \Lambda$ , \\
\quad 0         & else .
}
\end{equation}
Here we have used that $\Lambda$ is connected.
\end{proof}
\begin{lemma} \label{lemma:finitness1}
Let $n \in \NN$, $\Theta = \{0, \dots , n-1 \}$, $s \in (0,1)$, and $\Gamma \subset \ZZ$ be connected. Then,
\begin{enumerate}[(i)]
\item for every pair $x,x+n-1 \in \Gamma$ and all $z \in \CC \setminus \RR$ we have
\begin{equation} \label{eq:finitness1}
 \mathbb{E}_{\{x\}} \bigl \{ | G_\Gamma (z;x,x+n-1)|^{s/n}  \bigr \} \leq  C_u C_\rho =: C_{u,\rho} .
\end{equation}
\item if $1 \leq \abs{\Gamma} \leq n$, we have for  all $z \in \CC \setminus \RR$ the bound
\begin{equation} \label{eq:finitness2}
 \mathbb{E}_{\{\gamma_0\}} \bigl \{ | G_\Gamma (z;\gamma_0,\gamma_1)|^{s/n}  \bigr \} \leq C_u^+ C_\rho^+ =: C_{u,\rho}^+
\end{equation}
where $\gamma_0 = \min \Gamma$ and $\gamma_1 = \max \Gamma$.
\item if $\Gamma = \{x,x+1,...\}$ and $y \in \Gamma$ with $0 \leq y - x \leq n-1$, we have for all $z \in \CC \setminus \RR$ the bound
\begin{equation} \label{eq:finitness3}
 \mathbb{E}_{\{y-n+1\}}\bigl\{| G_\Gamma (z;x,y)|^{s/n}\bigr\} \leq C_{u,+} C_{\rho}^+ =: C_{u,\rho,+} .
\end{equation}
\end{enumerate}
The constants $C_u$, $C_\rho$, $C_u^+$, $C_\rho^+$ and $C_{u,+}$ are given in Eq.~\eref{eq:cucrho}, \eref{eq:cucrhoplus} and \eref{eq:cucrhominus}.
\end{lemma}
\begin{proof}
We start with the first statement of the lemma. By assumption $x,x+n-1 \in \Gamma$. We apply Lemma \ref{lemma:fraction} with $\Lambda := \{x, x+1, \dots, x+n-1\}\subset \Gamma$ (since $\Gamma$ is connected) and obtain for all $x,y \in \Lambda$
\[
 G_\Gamma (z;x,y) = \bigl \langle \delta_x , (H_{\Lambda} - B_\Gamma^\Lambda - z)^{-1} \delta_{y} \bigr \rangle ,
\]
where the operator $B_\Gamma^\Lambda$ is given by Eq.{} \eref{eq:bij}.
Set $D = H_{\Lambda} - B_\Gamma^\Lambda - z$.
By Cramer's rule we have $G_\Gamma (z;x,y) =  \det C_{y,x} / \det [D]$. Here, ${C}_{i,j} = (-1)^{i+j} M_{i,j}$ and $M_{i,j}$ is obtained from the tridiagonal matrix $[D]$ by deleting row $i$ and column $j$. Thus $C_{x+n-1,x}$ is a lower triangular matrix with determinant $\pm 1$. Hence,
\[
 \abs{G_\Gamma (z;x,x+n-1)} = \frac{1}{\abs{\det [D]}}.
\]
Since $\Theta = \supp u = \{0,\dots,n-1 \}$, every potential value $V_\omega (k)$, $k \in \Lambda$, depends on the random variable $\omega_{x}$, while the operator $B_\Gamma^\Lambda $ is independent of $\omega_x$. Thus we may write $[D]$ as a sum of two matrices
\[
 [D] = A + \omega_{x} V,
\]
where $V \in \RR^{n \times n}$ is diagonal with the elements $u(k-x)$, $k=x,\dots,x+n-1$, and $A := [D] - \omega_x V$. Since $A$ is independent of $\omega_x$ we may apply Lemma \ref{lemma:det} and obtain for all $s \in (0,1)$ the estimate \eref{eq:finitness1} with
\begin{equation} \label{eq:cucrho}
 C_u = \Bigl | \prod_{k \in \Theta} u(k) \Bigr |^{-s/n} \quad \mbox{and} \quad C_\rho =
 \|\rho\|_{\infty}^{s} \frac{2^{s} s^{-s}}{1-s} .
\end{equation}
The proof of Ineq.{} \eref{eq:finitness2} is similar but does not require Lemma \ref{lemma:fraction}. We have the decomposition  $[H_\Gamma - z] = \tilde A + \omega_{\gamma_0} \tilde V$, where
$d := \gamma_1-\gamma_0$,  $\tilde{V}\in \RR^{(d+1) \times (d+1)}$ is diagonal with elements $u (k - \gamma_0)$, $k = \gamma_0,\dots,\gamma_1$, and $\tilde A :=[H_\Gamma - z] - \omega_{\gamma_0} \tilde V$ is independent of $\omega_{\gamma_0}$. By Cramer's rule and Lemma \ref{lemma:det} we obtain
\begin{equation*}
 \mathbb{E}_{\{\gamma_0\}} \bigl \{ | G_\Gamma (z;\gamma_0,\gamma_1)|^{t/(d+1)}  \bigr \} \leq
\Bigl | \prod_{k=0}^{d} u(k) \Bigr |^{-t/(d+1)} \|\rho\|_{\infty}^{t} \frac{2^{t} t^{-t}}{1-t}
\end{equation*}
for all $t \in (0,1)$. We choose $t = s \frac{d+ 1}{n}$ and obtain Ineq.{} \eref{eq:finitness2} with the constants
\begin{equation} \label{eq:cucrhoplus}
 C_u^+ = \max_{i \in \Theta} \Bigl | \prod_{k=0}^{i} u(k) \Bigr |^{-s/n} \quad \mbox{and} \quad C_\rho^+ = \frac{\max \bigl \{ \|\rho\|_{\infty}^{s} , \|\rho\|_{\infty}^{s/n} \bigr \} }{2^{-s} s^{s}(1-s)} .
\end{equation}
In the final step we have used $s \geq t$ and the monotonicity of $(0,1) \ni x \mapsto 2^x x^{-x}/(1-x)$. For the proof of the third statement we apply Lemma \ref{lemma:fraction} with $\Lambda = \{x,\dots,y\}$ and obtain using Cramer's rule $\abs{G_\Gamma (z;x,y)} = | 1 / \det [H_\Lambda - B_\Gamma^\Lambda - z]|$. Set $d:= y-x$. Notice that $B_\Gamma^\Lambda$ is independent of $\omega_{y-n+1}$, while every potential value $V_\omega (k)$, $k \in \Lambda$, depends on $\omega_{y-n+1}$. Thus we have the decomposition $[H_\Lambda - B_\Gamma^\Lambda - z] = A + \omega_{y-n+1} V$, where $V \in \RR^{(d+1)\times(d+1)}$ is diagonal with the elements $u(k)$, $k=n-1-d,\dots,n-1$, and $A:= [H_\Lambda - B_\Gamma^\Lambda - z] - \omega_{y-n+1} V$. Since $A$ is independent of $\omega_{y-n+1}$ we may apply Lemma~\ref{lemma:det} and obtain for all $t \in (0,1)$
\[
 \EE_{\{y-n+1\}} \Bigl\{ \bigl| G_\Gamma (z;x,y) \bigr|^{s/(d+1)} \Bigr\} \leq
 \Bigl| \prod_{k=n-1-d}^{n-1} u(k)\Bigr|^{-t/(d+1)} \| \rho \|_\infty^t \frac{2^t t^{-t}}{1-t} .
\]
We choose $t = s \frac{d+1}{n}$ and obtain Ineq~\eref{eq:finitness3} with
\begin{equation} \label{eq:cucrhominus}
 C_{u,+} := \max_{i \in \Theta} \Bigl|\prod_{k=n-1-i}^{n-1} u(k)\Bigr|^{-s/n} .
\end{equation}
Here we have used $s \geq t$ and the monotonicity of $(0,1) \ni x \mapsto 2^x x^{-x}/(1-x)$.
\end{proof}
%
%
%
%
%
%
\section{Exponential decay of Green's function} \label{sec:exp}
In this section we use so-called ``depleted'' Hamiltonians to formulate a geometric resolvent formula. Such Hamiltonians are obtained by setting to zero the ``hopping terms'' of the Laplacian along a collection of bonds. More precisely, let $\Lambda \subset \Gamma \subset \ZZ$ be arbitrary sets. We define the depleted Laplace operator $\Delta_\Gamma^\Lambda :\ell^2 (\Gamma) \to \ell^2 (\Gamma)$ by
\begin{equation*} \label{eq:de1} \fl
 \sprod{\delta_x}{\Delta_\Gamma^\Lambda \delta_y} :=
\cases{
  0 & if $x \in \Lambda$, $y \in \Gamma \setminus \Lambda$ or $y \in \Lambda$, $x \in \Gamma \setminus \Lambda$ , \\
  \bigl \langle \delta_x , \Delta_\Gamma \delta_y \bigr \rangle & else.
}
\end{equation*}
In other words, the hopping terms which connect $\Lambda$ with $\Gamma \setminus \Lambda$ or vice versa are deleted. The depleted Hamiltonian $H_\Gamma^\Lambda : \ell^2 (\Gamma) \to \ell^2 (\Gamma)$ is then defined by
\[
 H_\Gamma^\Lambda := -\Delta_\Gamma^\Lambda + V_\Gamma .
\]
Let further $T_\Gamma^\Lambda := \Delta_\Gamma - \Delta_\Gamma^\Lambda$ be the difference between the the ``full'' Laplace operator and the depleted Laplace operator.
Analogously to Eq.{} \eref{eq:greens} we use the notation $G_\Gamma^\Lambda (z) := (H_\Gamma^\Lambda - z)^{-1}$ and $G_\Gamma^\Lambda (z;x,y) := \bigl \langle \delta_x, G_\Gamma^\Lambda(z) \delta_y \bigr \rangle$. The second resolvent identity yields for arbitrary sets $\Lambda \subset \Gamma \subset \ZZ$
\begin{eqnarray}
 G_\Gamma (z)   & = G_\Gamma^\Lambda (z) + G_\Gamma (z) T_\Gamma^\Lambda G_\Gamma^\Lambda (z)             \label{eq:resolvent} \\[1ex]
        & = G_\Gamma^\Lambda (z) + G_\Gamma^\Lambda (z)T_\Gamma^\Lambda G_\Gamma (z) .
          \label{eq:resolvent2}
\end{eqnarray}
In the following we will use that $G_\Gamma^\Lambda (z;x,y) = G_\Lambda (z;x,y)$ for all $x,y \in \Lambda$, since $H_\Gamma^\Lambda$ is block-diagonal, and that $G_\Gamma^\Lambda (z;x,y) = 0$ if $x \in \Lambda$ and $y \not \in \Lambda$ or vice versa.
\begin{lemma} \label{lemma:iteration1}
 Let $n \in \NN$, $\Theta = \{0,\dots,n-1 \}$, $\Gamma \subset \ZZ$ be connected, and $s \in (0,1)$.
 Then we have for all $x,y \in \Gamma$ with $y-x \geq n$, $\Lambda := \{x+n,x+n+1,\dots\}\cap \Gamma$ and all $z \in \CC \setminus \RR$ the bound
\begin{equation*}
 \mathbb{E}_{\{x\}}\bigl\{| G_\Gamma (z;x,y)|^{s/n}\bigr\} \leq C_{u,\rho} \cdot | G_\Lambda (z;x+n,y)|^{s/n} .
\end{equation*}
In particular,
\begin{equation}\label{eq:iteration1}
 \mathbb{E} \bigl\{| G_\Gamma (z;x,y)|^{s/n}\bigr\} \leq C_{u,\rho} \cdot \mathbb{E} \bigl\{| G_\Lambda (z;x+n,y)|^{s/n}\bigr\} .
\end{equation}

\end{lemma}
\begin{proof}
Our starting point is Eq.{} \eref{eq:resolvent}. Taking the matrix element $(x,y)$ yields
\begin{equation*}
 G_\Gamma (z;x,y) = G_\Gamma^{\Lambda} (z;x,y) + \bigl \langle \delta_x , G_\Gamma (z) T_\Gamma^\Lambda G_\Gamma^\Lambda (z) \delta_y \bigr \rangle .
\end{equation*}
Since $x \not \in \Lambda$ and $y \in \Lambda$, the first summand on the right vanishes as the depleted Green's function $G_\Gamma^{\Lambda} (z;x,y)$ decouples $x$ and $y$. For the second summand we calculate
\begin{equation}\label{e:geometric-resolvent}
 G_\Gamma (z;x,y) = G_\Gamma (z;x,x+n-1) G_\Lambda (z;x+n,y)   .
\end{equation}
The second factor is independent of $\omega_x$. Thus, taking expectation with respect to $\omega_x$ bounds the first factor using Ineq.{} \eref{eq:finitness1} and the proof is complete.
\end{proof}
\begin{lemma} \label{lemma:iteration2}
 Let $n \in \NN$, $\Theta = \{0,\dots,n-1\}$, $\Gamma = \{x,x+1,...\}$, $y \in \Gamma$ with $n \leq y - x < 2n$, and $s \in (0,1)$. Then we have for all $z \in \CC \setminus \RR$ the bound
\begin{equation} \label{eq:iteration2}
 \mathbb{E}_{\{y-n+1,x\}}\bigl\{| G_\Gamma (z;x,y)|^{s/n}\bigr\} \leq C_{u,\rho}^+ C_{u,\rho} .
\end{equation}
\end{lemma}
\begin{proof}
The starting point is Eq.{} \eref{eq:resolvent2}. Choosing $\Lambda = \{x,\dots,y-n\}$ gives
\[
 G_\Gamma (z;x,y) = G_\Lambda (z;x,y-n) G_\Gamma (z;y-n+1,y) .
\]
Since $G_\Lambda (z;x,y-n)$ depends only on the potential values at lattice sites in $\Lambda$
it is independent of $\omega_{y-n+1}$. We take expectation with respect to $\omega_{y-n+1}$ to bound the second factor of the above identity using Ineq.{} \eref{eq:finitness1}.
Since $1 \leq \abs{\Lambda} \leq n$ by assumption, we may apply Ineq.{} \eref{eq:finitness2} to $G_\Lambda (z;x,y-n)$ which ends the proof.
\end{proof}
The proof of the following theorem will serve as a basis to complete the proof of
\begin{itemize}
 \item[(i)] Theorem \ref{thm:result1} at the end of this section,
 \item[(ii)] Theorem \ref{thm:result2} in Section~\ref{sec:gen}.
\end{itemize}
The difference between the proof of Theorem~\ref{thm:result1} and Theorem~\ref{theorem:exp} is, that the latter is better suited for a generalization to single-site potentials with disconnected support.
\begin{theorem} \label{theorem:exp}
Let $\Theta = \{0,\dots,n-1\}$, $\Gamma \subset \ZZ$ connected and $s \in (0,1)$. Assume
\begin{equation} \label{eq:disorder}
 \|\rho\|_\infty < \frac{(1-s)^{1/s}}{2s^{-1}} \Bigl | \prod_{k=0}^{n-1} u(k) \Bigr|^{1/n} .
\end{equation}
Then $m = - \ln C_{u,\rho}$ is strictly positive and
\[
 \mathbb{E} \bigl\{ | G_\Gamma (z;x,y)|^{s/n} \bigr\} \leq C_{u,\rho}^+ \,\, \exp \Biggl\{-m \Biggl \lfloor \frac{| x-y|}{n} \Biggl\rfloor\Biggr\}
\]
for all $x,y \in \Gamma$ with $| x-y | \geq 2n$ and all $z \in \CC \setminus \RR$. Here, $\lfloor \cdot \rfloor$ is defined by $\lfloor z \rfloor := \max\{k\in \ZZ \mid k\leq z\}$.
\end{theorem}
\begin{proof}
The constant $m$ is larger than zero since $C_{u,\rho} < 1$ by assumption.
By symmetry we assume without loss of generality $y-x \geq 2n$.
In order to estimate $\mathbb{E} \bigl \{| G_\Gamma (z;x,y)|^{s/n}\bigr\}$,
we iterate Eq.{} \eref{eq:iteration1} of Lemma \ref{lemma:iteration1}
and finally use Eq.{} \eref{eq:iteration2} of Lemma \ref{lemma:iteration2} for the last step.
Figure \ref{fig:iteration} shows this procedure schematically.
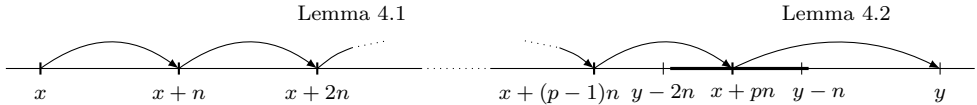
\begin{figure}[ht]
\centering
\begin{tikzpicture}[scale=0.92]
 \draw[very thick] (9.1,0) -- (11.1,0);
 \draw (-0.5,0)--(5.5,0);\draw[dotted] (5.5,0)--(6.5,0);\draw (6.5,0)--(13.5,0);
 \draw[thick] (0 cm,2.5pt) -- (0 cm,-2.5pt) node[anchor=north]  {\footnotesize $\phantom{I}x\phantom{I}$};
 \draw[thin] (13 cm,2.5pt) -- (13 cm,-2.5pt) node[anchor=north]{\footnotesize $\phantom{I}y\phantom{I}$};
 \draw[thick] (2 cm,2.5pt) -- (2 cm,-2.5pt) node[anchor=north]  {\footnotesize $\phantom{I}x+n\phantom{I}$};
 \draw[thick] (4 cm,2.5pt) -- (4 cm,-2.5pt) node[anchor=north]  {\footnotesize $\phantom{I}x+2n\phantom{I}$};
 \draw[thin] (11 cm,2.5pt) -- (11 cm,-2.5pt);
 \draw (11.25,-0.07) node[anchor=north]{\footnotesize $\phantom{I}y-n\phantom{I}$};
 \draw[thin] (9 cm,2.5pt) -- (9 cm,-2.5pt) node[anchor=north]  {\footnotesize $\phantom{I}y-2n\phantom{I}$};
 \draw[thick] (10 cm,2.5pt) -- (10 cm,-2.5pt);
 \draw (10.15,-0.07) node[anchor=north]{\footnotesize $\phantom{I}x+pn\phantom{I}$};
 \draw[thick] (8 cm,2.5pt) -- (8 cm,-2.5pt);
\draw (7.5,-0.07) node[anchor=north]{\footnotesize $\phantom{I}x+(p-1)n\phantom{I}$};
 \draw[-latex] (0,0) .. controls  (0.7,0.5) and (1.3,0.5) .. (2,0);  
 \draw[-latex] (2,0) .. controls  (2.7,0.5) and (3.3,0.5) .. (4,0);
 \draw         (4,0) .. controls  (4.25,0.2)  .. (4.5,0.3);
\draw[-latex] (8,0) .. controls  (8.7,0.5) and (9.3,0.5) .. (10,0);
 \draw[dotted] (4.5,0.3) .. controls  (5,0.4)  .. (5,0.4);
 \draw[dotted] (7,0.4) .. controls (7,0.4) .. (7.5,0.3);
 \draw[-latex]        (7.5,0.3) .. controls  (7.75,0.20)  .. (8,0);
\draw[-latex] (10,0) .. controls  (11,0.5) and (12,0.5)  .. (13,0);
\draw (4.5,0.8) node {\footnotesize Lemma \ref{lemma:iteration1}};
\draw (11.5,0.8) node {\footnotesize Lemma \ref{lemma:iteration2}};
\end{tikzpicture}
\caption{Illustration to the proof of Theorem \ref{theorem:exp}}
\label{fig:iteration}
\end{figure}
We choose $p := \lfloor (y-x)/n \rfloor - 1 \in \NN$ such that $y-2n < x+pn \leq y-n$.
We iterate Eq.{} \eref{eq:iteration1} exactly $p$ times and obtain
\[
 \mathbb{E} \bigl\{| G_\Gamma (z;x,y)|^{s/n}\bigr\} \leq C_{u,\rho}^{p} \cdot \mathbb{E}\bigl\{| G_{\Lambda_{p}} (z;x+pn,y)|^{s/n}\bigr \}
\]
where $\Lambda_p = \{ x+pn, x+pn+1, \dots \}$. Now the first $p$ jumps of Fig. \ref{fig:iteration} are done
and it remains to estimate $\mathbb{E}\bigl\{| G_{\Lambda_{p}} (z;x+pn,y)|^{s/n}\bigr\}$.
Since $n \leq y-(x+pn) < 2n$ and $\Lambda_p = \{ x+pn, x+pn+1, \dots \}$ we may apply Lemma \ref{lemma:iteration2} and get
\[
 \mathbb{E}\bigl\{| G_\Gamma (z;x,y)|^{s/n}\bigr\} \leq C_{u,\rho}^{p+1} C_{u,\rho}^+
= C_{u,\rho}^+ \,\, {\rm e}^{(p+1) \ln C_{u,\rho}} . \qedhere
\]
\end{proof}
\begin{proof}[Proof of Theorem \ref{thm:result1}]
Without loss of generality we assume $y-x \geq n$. We iterate Eq.{} \eref{eq:iteration1} exactly $q := \lfloor (y-x)/n \rfloor \in \NN$ times, starting with $\Gamma = \ZZ$, and obtain $\mathbb{E} \bigl\{| G_\omega (z;x,y)|^{s/n}\bigr\} \leq C_{u,\rho}^{q} \cdot \mathbb{E}\bigl\{| G_{\Lambda_{q}} (z;x+qn,y)|^{s/n}\bigr \}$, where $\Lambda_q = \{ x+pn, x+pn+1, \dots \}$. Since $0 \leq y-(x+qn) \leq n-1$ by construction, we may apply part (iii) of Lemma \ref{lemma:finitness1} and obtain
\begin{equation} \label{eq:Cm}
  \mathbb{E}\bigl\{| G_\omega (z;x,y)|^{s/n}\bigr\} \leq C_{u,\rho}^{q} C_{u,\rho,+}
= C_{u,\rho,+} \,\, \exp \Biggl\{-m \biggl\lfloor \frac{y-x}{n}\biggr\rfloor\Biggr\}
\end{equation}
where $m = -\ln C_{u,\rho}$. In particular, $m>0$ if Ineq. \eref{eq:disorder} holds.
\end{proof}
%
%
%
%
%
%
%
%
\section{Single-site potentials with arbitrary finite support} \label{sec:gen}
In this section we prove Theorem \ref{thm:result2}. We consider the case in which the support $\Theta$ of the single-site potential is an arbitrary finite subset of $\ZZ$.  By translation, we assume without loss of generality that $\min \Theta = 0$ and $\max \Theta = n-1$ for some $n \in \NN$. Furthermore, we define
\begin{equation} \label{eq:r}
 r := \max\big\{\, b-a \, \mid [a,b] \subset \{0,\dots,n-1\}, [a,b] \cap \Theta = \emptyset\big\} .
\end{equation}
Thus $r$ is the width of the largest gap in $\Theta$. In order to handle arbitrary finite supports of the single-site potential, we need one of the following additional assumptions on the density $\rho \in L^{\infty} (\RR)$:
\begin{equation} \label{eq:assumptions}
 \mathcal{A}_1: \rho \in W^{1,1} (\RR) \qquad \mathcal{A}_2: \supp \rho \subset [-R,R] \mbox{ for some $R>0$}.
\end{equation}
To illustrate the difficulties arising for non-connected supports $\Theta$ we consider an example.
Suppose $\Theta = \{0,2,3,\dots , n-1\}$ so that $r = 1$. If we set $\Lambda = \{0,\dots,n-1\}$ there is no decomposition $H_\Lambda - B_\Gamma^\Lambda = A + \omega_0 V$ with an invertible $V$. If we set $\Lambda = \{0,\dots,n-1+r\} = \{0,\dots,n\}$ we observe that every diagonal element of $H_\Lambda$ depends at least on one of the variables $\omega_0$ and $\omega_1 = \omega_r$, while the elements of $B_\Gamma^\Lambda$ (which appear after applying Lemma \ref{lemma:fraction}) are independent of $\omega_k$, $k \in \{0,\dots,r\} = \{0,1\}$. Thus we have a decomposition $H_\Lambda - B_\Gamma^\Lambda = A + \omega_0 V_0 + \omega_1 V_1$, where $A$ is independent of $\omega_k$, $k \in \{0,1\}$, and for all $i \in \Lambda$ either $V_0 (i)$ or $V_1 (i)$ is not zero.
As a consequence there is an $\alpha \in \RR$ such that $V_0 + \alpha V_1$ is invertible on $\ell^2 (\Lambda )$. Motivated by this observation, we prove the following lemma.
\begin{lemma} \label{lemma:detgen}
Let $N,d \in \NN$ and $A, V_0, V_1, \dots,V_N \in \CC^{d \times d}$ be matrices.
Let $(\alpha_k)_{k=0}^N \in \RR^{N+1}$ with $\alpha_0 \not = 0$. Assume that $\sum_{k=0}^N \alpha_k V_k$ is invertible. Let further $0 \leq \rho \in L^1(\RR) \cap L^\infty (\RR)$ with $\|\rho\|_{L^1} = 1$, $t \in (0,1)$, and $\mathcal{A}_1,\mathcal{A}_2$ be as in \eref{eq:assumptions}.
Then, if the condition $\mathcal{A}_1$ is satisfied, we have the bound
\begin{eqnarray*} \fl
I &= \!\!\!\!\! \int\limits_{\RR^{N+1}} \!\!\! \Bigl|\det \Bigl(A + \sum_{i=0}^N r_iV_i\Bigr)\Bigl|^{\frac{t}{d}} \prod_{i=0}^N\rho (r_i) \drm r_i 
&\leq \Bigl| \det \Bigl( \sum_{k=0}^N \alpha_k V_k \Bigr)\Bigr|^{-\frac{t}{d}} \Bigl(\sum_{k=0}^N | \alpha_k| \Bigr)^t \,\, \frac{t^{-t}}{1-t}  \|\rho'\|_{L^1}^{t} .
\end{eqnarray*}
If the condition $\mathcal{A}_2$ is satisfied, we have the bound
\begin{equation*} \fl
I  \leq \Bigl| \det \Bigl( \sum_{k=0}^N \alpha_k V_k \Bigr)\Bigr|^{-t/d} | \alpha_0|^t\Bigl(1+ \max_{i\in \{1,\dots,N\}} \frac{|\alpha_i|}{| \alpha_0|} \Bigr)^{Nt}
\,\, \frac{2^t t^{-t}}{1-t} (2R)^{Nt} \|\rho\|_\infty^{(N+1)t} .
\end{equation*}
\end{lemma}
\begin{proof}
Substituting
\[ \fl
 \pmatrix{
  r_0 \cr
  r_1 \cr
  \vdots \cr
  \vdots \cr
  r_N} =
T \pmatrix{
  x_0 \cr
  x_1 \cr
  \vdots \cr
  \vdots \cr
  x_N \cr}
=
 \pmatrix{
  \alpha_0 & 0        & \cdots  &\cdots   &  0     \cr
  \alpha_1 & \alpha_0 &0        &         & \vdots \cr
  \alpha_2 & 0        & \alpha_0&\ddots   & \vdots \cr
  \vdots   & \vdots   & \ddots  &\alpha_0 & 0      \cr
  \alpha_N & 0        &         & 0       & \alpha_0 \cr}
 \pmatrix{
  x_0 \cr
  x_1 \cr
  \vdots \cr
  \vdots \cr
  x_N \cr }
=
\pmatrix{
  \alpha_0 x_0   \cr
  \alpha_1 x_0 + \alpha_0 x_1 \cr
  \alpha_2 x_0 + \alpha_0 x_2 \cr
  \vdots     \cr
  \alpha_N x_0 + \alpha_0 x_N \cr}
\]
we get
\[ \fl
 I = \int_{\RR^{N}} \left(  \int_\RR \Bigl|\det \Bigl(\tilde A +  x_0\sum_{i=0}^N \alpha_i V_i\Bigr)\Bigr|^{-t/d} g(x_0,\dots,x_N) \drm x_0 \right) | \alpha_0|^{N+1} \drm x_1  \dots  \drm x_N
\]
where $\tilde A = A + \alpha_0 \sum_{i=1}^N x_i V_i$ and
$g(x_0,\dots,x_N) = \rho(\alpha_0 x_0)\prod_{i=1}^N \rho(\alpha_i x_0 + \alpha_0 x_i)$.
Since $x_0 \mapsto g(x_0,\dots,x_N)$ is an element of $L^1(\RR) \cap L^\infty (\RR)$ we may apply Lemma \ref{lemma:det} and obtain for all $\lambda > 0$
\begin{eqnarray*} \fl
I \leq
\Bigl|\det \Bigl(\sum_{i=0}^N \alpha_i V_i\Bigr)\Bigr|^{-t/d} \Bigl(
   \lambda^{-t}  + \frac{2 \lambda^{1-t}}{1-t} \int_{\RR^N} \sup_{x_0 \in \RR} g(x_0,\dots,x_N) |\alpha_0|^{N+1}  \drm x_1   \dots  \drm x_N \Bigr)
\end{eqnarray*}
where $\drm x = \drm x_1  \dots  \drm x_N$. In the case of $\mathcal{A}_1$ we use $\sup_{x_0 \in \RR} g \leq \frac{1}{2} \int_\RR |\partial g/\partial x_0| \drm x_0$, substitute back into the original coordinates and finally choose $\lambda = t/(\|\rho'\|_{L^1} \sum_{k=0}^{N} |\alpha_k|)$.
To end the proof if the condition $\mathcal{A}_2$ is satisfied, we use $\supp \rho \subset [-R,R]$ and see that if $| x_j| > R \,\| T^{-1}\|_\infty $
for some $j = 0,\dots,N$, then $g(x_0,\dots,x_N) = 0$. Thus it is sufficient to integrate over the cube $[-R\| T^{-1}\|_\infty,R\| T^{-1}\|_\infty]^N$.
We estimate $\sup_{x_0 \in \RR} g(x_0,\dots,x_N) \leq \| \rho \|_\infty^{N+1}$ and choose $\lambda = t/(2 \|\rho\|_\infty^{N+1} |\alpha_0^{N+1}| (2R\| T^{-1}_\infty \|)^N)$.
The row-sum norm of $T^{-1}$ equals $\| T^{-1} \|_\infty = \max_{i \in \{1,\dots,N\}}$ $(\abs{\alpha_0}^{-1} + |\alpha_i/\alpha_0^2|) = (1+\max_{i\in\{1,\dots,N\}} |\alpha_i/\alpha_0|)/|\alpha_0|$.
\end{proof}
With the help of Lemma \ref{lemma:detgen} we prove the following analogues of Lemma \ref{lemma:finitness1} and Theorem \ref{theorem:exp}.
\begin{lemma} \label{lemma:finitness3}
 Let $n \in \NN$, $\Theta \subset \ZZ$ with $\min \Theta = 0$, $\max \Theta = n-1$, and $\Gamma \subset \ZZ$ be connected. Let further $r$ be as in Eq.{} \eref{eq:r}, $\mathcal{A}_1, \mathcal{A}_2$ as in \eref{eq:assumptions},
 and $s \in (0,1)$.
Then there exists a constant $D$ such that for all $x,x+n-1+r \in \Gamma$ and $z \in \CC \setminus \RR$
\begin{equation} \label{eq:leqD}
\mathbb{E}_{\{x,\dots,x+r\}} \bigl \{ | G_\Gamma (z;x,x+n-1+r)|^{s/(n+r)}  \bigr \} \leq D\, .
\end{equation}
 The constant $D$ is characterized in Eq.{} \eref{eq:defD} and estimated in Ineq.{} \eref{eq:D1} and \eref{eq:D2}. If $1 \leq \abs{\Gamma} \leq n+r$ with $\gamma_0 = \min \Gamma$ and $\gamma_1 = \max \Gamma$ there exists a constant $D^+$ such that for all $z \in \CC \setminus \RR$
\begin{equation} \label{eq:leqDplus}
 \mathbb{E}_{\{\gamma_0,\dots,\gamma_0+r\}} \bigl \{ | G_\Gamma (z;\gamma_0,\gamma_1) |^{s/(n+r)}  \bigr \} \leq D^+.
\end{equation}
The constant $D^+$ is characterized in Eq.{} \eref{eq:defDplus} and estimated in Ineq.{} \eref{eq:Dplus1} and \eref{eq:Dplus2}.
\end{lemma}
\begin{proof}
 The proof is similar to the proof of Lemma \ref{lemma:finitness1}. Apply Lemma \ref{lemma:fraction} with $\Lambda = \{x,x+1,\dots,x+n-1+r\}$ and Cramer's rule to get $| G_\Gamma (z;x,x+n-1+r)| = 1/\abs{\det [D]}$ where $D = H_\Lambda - B_\Gamma^\Lambda - z$. Note that $B_\Gamma^\Lambda$ is independent of $\omega_k$, $k \in \{x,\dots,x+r\}$. We have the decomposition $[D] = A + \sum_{k=0}^{r} \omega_{x+k} V_k$ where the elements of the diagonal matrices $V_k \in \RR^{(n+r) \times (n+r)}$, $k = 0,\dots,r$, are given by $V_k(i) = u(i-k)$, $i=0,\dots,n-1+r$, and $A = D - \sum_{k=0}^r \omega_k V_k$ is independent of $\omega_k$, $k \in \{x,\dots,x+r\}$.
We apply Lemma \ref{lemma:detgen} and obtain for all $\alpha = (\alpha_k)_{k=0}^r \in M := \{\alpha \in \RR^{r+1} : \alpha_0 \not = 0,  \mbox{ $\sum_{k=0}^r \alpha_k V_k$ is invertible}\}$ the bound  $\mathbb{E}_{\{x,\dots,x+r\}} \bigl \{ | G_\Gamma (z;x,x+n-1+r)|^{s/(n+r)} \bigr \} \leq D_\alpha$ where
\begin{equation*} 
 D_\alpha  = \|\rho'\|_{L^1}^{s} \frac{s^{-s}}{1-s}   \Bigl(\sum_{k=0}^r \abs{\alpha_k} \Bigr)^s \,\, \prod_{i=0}^{n-1+r} \Bigl| \sum_{k=0}^r \alpha_k u(i-k) \Bigr|^{-\frac{s}{n+r}}
\end{equation*}
if $\mathcal{A}_1$ is satisfied and
\begin{equation*} \fl
 D_\alpha  = \|\rho\|_\infty^{(r+1)s} (2R)^{rs} \frac{2^s s^{-s}}{1-s} | \alpha_0|^s \Biggl( 1+\max_{i\in \{1,\dots,r\}} \frac{|\alpha_i|}{|\alpha_0|} \Biggr)^{rs} \, \prod_{i=0}^{n-1+r} \Bigl| \sum_{k=0}^r \alpha_k u(i-k) \Bigr|^{-\frac{s}{n+r}}
\end{equation*}
if $\mathcal{A}_2$ is satisfied. The set $M$ is non-empty and equal to the set $\{\alpha \in \RR^{r+1} : \alpha_0 \not = 0, D_\alpha \mbox{ is finite}\}$.
Thus Ineq.{} \eref{eq:leqD} holds with the constant
\begin{equation} \label{eq:defD}
 D := \inf_{\alpha \in M} D_\alpha .
\end{equation}
In the following we establish an upper bound for $D$. Using a volume comparison criterion we can find a vector $\alpha'=(\alpha_k')_{k=0}^r \in [0,1]^{r+1}$ which has to each hyperplane $\sum_{k=0}^r \alpha_k u(i-k) = 0$, $i=0,\dots,n-1+r$, at least the Euclidean distance $(2 (n+r) (r+1)^{r/2})^{-1}$, as outlined in Fig. \ref{fig:volume}.
\begin{figure}[t]
\centering
\begin{tikzpicture}[scale=4]
\draw[dotted] (0,1)--(1,1);\draw[dotted] (1,1)--(1,0);\draw[dotted] (1,0)--(0,0);\draw[dotted] (0,0)--(0,1);
\draw[<->] (1.1207,0.9793)--(0.9793,1.1207); \draw (1.1,1.1) node {\small $\epsilon$};
\draw (0,0)--(0,1);\draw[dashed] (0.1,0)--(0.1,1);
\draw (0,0)--(1,1); \draw[dashed] (0,0.14)--(0.86,1); \draw[dashed] (0.14,0)--(1,0.86);
\draw (0,0)--(1,0.4); \draw[dashed] (0,0.108)--(1,0.508); \draw[dashed] (0.27,0)--(1,0.292);
\draw (0,1.1) node {\small $H_0^\epsilon$};
\draw (1.1,0.8) node {\small $H_1^\epsilon$};
\draw (1.185,0.4) node {\small $H_{n-1+r}^\epsilon$};
\filldraw (0.4,0.95) circle (0.4pt); \draw (0.4,0.9)  node {\small $\alpha'$};
\draw (2.3,0.8) node {\small $\mathrm{Vol}\,(W) = 1$};
\draw (2.3,0.6) node {\small $\mathrm{Vol}\,(\cup_i H_i^\epsilon) \leq (n+r) (r+1)^{r/2} \epsilon$};
\draw (2.3,0.4) node {\small $\mathrm{Vol}\,(W \setminus \cup_i H_i^\epsilon) \geq 1 - (n+r) (r+1)^{r/2} \epsilon$};
\end{tikzpicture}
\caption{Sketch of the existence of a vector $\alpha' \in W = [0,1]^{r+1}$
with the desired properties: Let  $H_i^\epsilon$ denote the $\epsilon$-neighborhood of the hyperplane $H_i:=\{ \alpha \in W \mid \sum_{k=0}^r \alpha_k u(i-k) = 0\}$
for  $i \in \{0,\dots,n-1+r\}$. Since the volume of $W \setminus \cup_i H_i^\epsilon$ is positive if $\epsilon$ is smaller  than $(n+r)^{-1} (r+1)^{-r/2} = d_0$, we conclude (using continuity) that there is a vector $\alpha'$ whose distance to each hyperplane $H_i$, $i \in \{0,\dots,n-1+r\}$, is at least $d_0/2$.}
\label{fig:volume}
\end{figure}
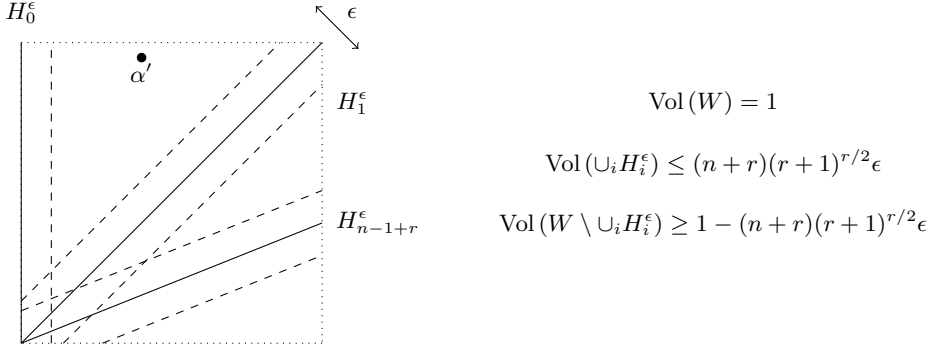
This implies $\alpha_0' \geq ( 2 (n+r) (r+1)^{r/2})^{-1}$ since the hyperplane for $i=0$ is $\alpha_0 = 0$. With this choice of $\alpha$ and the notation $u_i = (u(i-k))_{k=0}^r$, $i \in \{0,\dots,n-1+r\}$, we have
\begin{eqnarray} 
\prod_{i=0}^{n-1+r} \Bigl| \sum_{k=0}^r \alpha_k' u(i-k) \Bigr|^{-\frac{s}{n+r}}
&=\prod_{i=0}^{n-1+r} \Bigl| \| u_i \| \sprod{\alpha'}{u_i/ \| u_i\|}_2 \Bigr|^{-\frac{s}{n+r}} \nonumber \\
&\leq
\frac{\bigl[ 2 (n+r) (r+1)^{r/2} \bigr]^s}{\Bigl| \prod_{i=0}^{n-1+r} \Bigl( \sum_{k=0}^r u(i-k)^2 \Bigr) \Bigr|^{\frac{s}{2(n+r)}}} \label{eq:volume}
\end{eqnarray}
where $\sprod{\cdot}{\cdot}_2$ denotes the standard Euclidian scalar product. Now, in both cases $\mathcal{A}_1$ and $\mathcal{A}_2$ we choose $\alpha = \alpha'$ and obtain
\begin{equation} \label{eq:D1}
 D \leq \|\rho'\|_{L^1}^{s} \frac{s^{-s}}{1-s} \frac{(r+1)^s \bigl[ 2 (n+r) (r+1)^{r/2} \bigr]^s}{\Bigl| \prod_{i=0}^{n-1+r} \Bigl( \sum_{k=0}^r u(i-k)^2 \Bigr) \Bigr|^{\frac{s}{2(n+r)}}}
\end{equation}
if $\mathcal{A}_1$ is satisfied and
\begin{equation} \label{eq:D2} \fl
 D \leq  \|\rho\|_\infty^{(r+1)s}  \frac{2^s s^{-s}}{1-s} \bigl( 1+2 (n+r) (r+1)^{r/2} \bigr)^{rs} \frac{(2R)^{rs} \bigl[ 2 (n+r) (r+1)^{r/2} \bigr]^s}{\Bigl| \prod_{i=0}^{n-1+r} \Bigl( \sum_{k=0}^r u(i-k)^2 \Bigr) \Bigr|^{\frac{s}{2(n+r)}}}
\end{equation}
if $\mathcal{A}_2$ is satisfied.
\par
The proof of the second statement is similar but without use of Lemma \ref{lemma:fraction}. By Cramer's rule we get $| G_\Gamma (z;\gamma_0,\gamma_1)| = 1/|\det [H_\Gamma - z]|$. Set $d = \gamma_1 - \gamma_0$. We have the decomposition $[H_\Gamma - z] = \tilde A + \sum_{k=0}^{r} \omega_{\gamma_0+k} \tilde V_k$, where the elements of the diagonal matrices $\tilde V_k \in \RR^{(d + 1) \times (d + 1)}$, $k = 0,\dots,r$, are given by $\tilde V_k(i) = u(i-k)$, $i  \in \{0,\dots,d\}$, and $\tilde A := [H_\Gamma - z] - \sum_{k=0}^r \omega_k \tilde V_k$ is independent of $\omega_k$, $k \in \{x,\dots,x+r\}$.
We apply Lemma \ref{lemma:detgen} with $t = s \frac{d+1}{n+r}$ and obtain (using $s \geq t$) for all
$\alpha = (\alpha_k)_{k=0}^r \in \tilde M := \{\alpha \in \RR^{r+1} : \alpha_0 \not = 0,~ \mbox{$\sum_{k=0}^r \alpha_k \tilde V_k$ is invertible}\}$ that $\mathbb{E}_{\{\gamma_0,\dots,\gamma_0+r\}} \bigl \{ | G_\Gamma (z;\gamma_0,\gamma_1) |^{s/(n+r)} \bigr \} \leq D_\alpha^+ (d)$ where
\[
 D_\alpha^+ (d)  = \|\rho'\|_{L^1}^{s \frac{d+1}{n+r}} \frac{s^{-s}}{1-s}   \Bigl(\sum_{k=0}^r | \alpha_k| \Bigr)^{s \frac{d+1}{n+r}}  \prod_{i=0}^{d} \Bigl| \sum_{k=0}^r \alpha_k u(i-k) \Bigr|^{-s/(n+r)}
\]
if $\mathcal{A}_1$ is satisfied and $D_\alpha^+ (d)$ equals
\begin{equation*} \fl
    \|\rho\|_\infty^{s \frac{(r+1)(d+1)}{n+r}} (2R)^{s \frac{r(d+1)}{n+r}} \frac{2^s s^{-s}}{1-s} |\alpha_0|^{s \frac{d+1}{n+r}}\Biggl( 1+\max_{i\in \{1,\dots,r\}} \frac{|\alpha_i|}{|\alpha_0|} \Biggr)^{s r} \prod_{i=0}^{d} \Bigl| \sum_{k=0}^r \alpha_k u(i-k) \Bigr|^{\frac{s}{n+r}}
\end{equation*}
if $\mathcal{A}_2$ is satisfied. Since $\tilde M \supset M$ for each $d \in {0,\dots n-1+r}$ the set $\tilde M$ is non-empty. Thus Ineq.{} \eref{eq:leqDplus} holds with the constant
\begin{equation} \label{eq:defDplus}
 D^+ := \max_{d \in \{0,\dots,n-1+r\}\phantom{\tilde M}} \inf_{\alpha \in \tilde M} D_\alpha^+ (d) .
\end{equation}
We again choose $\alpha = \alpha'$ as in Fig. \ref{fig:volume}, use $\alpha_k' \in [0,1]$ and $\alpha_0' \geq (2 (n+r) (r+1)^{r/2})^{-1}$, estimate $D_{\alpha'}^+ (d)$ similar to Ineq.{} \eref{eq:volume}, and obtain
\begin{equation} \label{eq:Dplus1} \fl
 D^+ \leq  \max_{d \in \{0,\dots,n-1+r\}} \left\lbrace   \|\rho'\|_{L^1}^{s \frac{d+1}{n+r}} \frac{s^{-s}}{1-s}    \frac{(r+1)^s \bigl[2(d+1)(r+1)^{r/2} \bigr]^s}{\Bigl| \prod_{i=0}^{d} \sum_{k=0}^r u(i-k)^2 \Bigr|^{s/(2(n+r))}} \right\rbrace
\end{equation}
if $\mathcal{A}_1$ is satisfied and
\begin{eqnarray} \label{eq:Dplus2} \fl
D^+\leq \\ \fl \max_{d \in \{0,\dots,n-1+r\}} 
\left\lbrace
 \frac{\|\rho\|_\infty^{s\frac{(r+1)(d+1)}{n+r}}}{(2R)^{-s\frac{r(d+1)}{n+r}}} \cdot \frac{\bigl[1+2 (d+1) (r+1)^{r/2} \bigr]^{sr} \bigl[2(d+1)(r+1)^{r/2} \bigr]^s}{2^{-s} s^s (1-s) \Bigl| \prod_{i=0}^d \sum_{k=0}^r u(i-k) \Bigr|^{\frac{s}{2(n+r)}}}
\right\rbrace \nonumber
\end{eqnarray}
if $\mathcal{A}_2$ is satisfied.
\end{proof}
\begin{theorem} \label{theorem:exp2}
Let $n \in \NN$, $\Theta \subset \ZZ$, $\min \Theta = 0$, $\max \Theta = n-1$, $\Gamma \subset \ZZ$ connected, $s \in (0,1)$, $r$ as in Eq.{} \eref{eq:r}, $D$ the constant from Lemma \ref{lemma:finitness3}, and let $\rho$ satisfy one of the assumptions $\mathcal{A}_1$ or $\mathcal{A}_2$ from \eref{eq:assumptions}. Assume $D < 1$.
Then $m = - \ln D$ is strictly positive and we have the bound
\[
 \mathbb{E} \bigl\{ | G_\Gamma (z;x,y)|^{s/(n+r)} \bigr\} \leq D^+ \,\, {\rm e}^{-m \left \lfloor \frac{| x-y|}{n+r} \right\rfloor}
\]
for all $x,y \in \ZZ$ with $| x-y|\geq 2(n+r)$ and all $z \in \CC \setminus \RR$, where $\lfloor \cdot \rfloor$ is defined by $\lfloor z \rfloor:=\max\{k\in \ZZ \mid k\leq z\}$.
\end{theorem}
\begin{proof}
 The proof is similar to the proof of Theorem \ref{theorem:exp}. We again assume $y>x$. Let $\Gamma_1 \subset \ZZ$ be connected. Using Eq.{} \eref{eq:resolvent} with $\Lambda := \{x+n+r,\dots\} \cap \Gamma_1$ and Lemma \ref{lemma:finitness3} we have for all pairs $x,y \in \Gamma_1$ with $y-x \geq n+r$
 \begin{equation} \label{eq:iteration3}
  \mathbb{E} \bigl\{ | G_{\Gamma_1} (z;x,y)|^{s/(n+r)}\bigr\} \leq D \,\,  \mathbb{E} \bigl\{ | G_\Lambda (z;x+n+r,y)|^{s/(n+r)}\bigr \}
\end{equation}
which is the analogue to Lemma \ref{lemma:iteration1}. Now, let $\Gamma_2 = \{x,x+1,\dots\}$ and $y \in \Gamma_2$ with $n+r \leq y-x < 2(n+r)$. By Eq.{} \eref{eq:resolvent2} with $\Lambda = \{x,\dots,y-(n+r)\}$ and Lemma \ref{lemma:finitness3} we have
\begin{equation} \label{eq:iteration4}
 \mathbb{E} \bigl\{ | G_{\Gamma_2} (z;x,y)|^{s/(n+r)}\bigr\} \leq D D^+
\end{equation}
    which is the analogue of Lemma \ref{lemma:iteration2}. Iterating Eq.{} \eref{eq:iteration3} exactly $p=\lfloor(y-x)/(n+r)\rfloor - 1$ times, starting with $\Gamma_1 = \Gamma$, and finally using Eq.{} \eref{eq:iteration4} once gives the statement of the theorem.
\end{proof}
\section{Apriori bound} \label{sec:apriori}
Here we prove a global uniform bound on $(x,y) \mapsto \EE\{ | G_\Gamma (z;x,y) |^s  \}$
for $s>0$ sufficiently small.
We assume throughout that assumption $\mathcal{A}_2$ holds, i.\,e.{} there is
an $R \in (0, \infty)$  such that $ \supp \rho \subset [-R, R]$.
We use the notation $u_j (x) = u(x-j)$,  for all $j,x \in \ZZ$,  for the translated function
as well as for the corresponding multiplication operator.
\begin{theorem}\label{t:a-priori}
Let $\Gamma \subset \ZZ$ be connected, $s \in (0,1)$, $\Theta \subset \ZZ$ with $\min \Theta = 0$ and $\max \Theta = n-1$ for some $n \in \NN$, and  $\supp \rho $ be compact.
Then there is a positive constant $C$ such that for all $x,y \in \Gamma$ and all $z \in \CC \setminus \RR$ we have
\[
\mathbb{E} \bigl \{ | G_\Gamma (z;x,y) |^{s/(4n)} \bigr \} \leq C .
\]
\end{theorem}
For the proof we will need
\begin{lemma}\label{lemma:averagenorm}
 Let $n \in \mathbb{N}$, $R \in \mathbb{R}$, $A \in \mathbb{C}^{n \times n}$ an arbitrary matrix, $V \in \mathbb{C}^{n \times n}$ an invertible matrix and $s \in (0,1)$. Then we have the bounds
\begin{equation} \label{eq:norm_estimate}
\| V^{-1} \| \leq \frac{\| V \|^{n-1}}{\abs{\det V}}
\end{equation}
and
\begin{equation} \label{eq:average_norm}
\int_{-R}^R \bigl\| (A+rV)^{-1} \bigr\|^{s/n} {\rm d}r \leq \frac{2R^{1-s} (\| A \| + R \| V \|)^{s(n-1)/n}}{s^s (1-s) \abs{\det V}^{s/n}} .
\end{equation}
\end{lemma}
\begin{proof}
To prove Ineq. \eref{eq:norm_estimate} let $0< s_1 \leq s_2 \leq \ldots \leq s_n$ be the singular values of $V$. Then we have $\prod_{i=1}^n s_i \leq s_1 s_n^{n-1}$, that is,
\begin{equation}\label{eq:2}
 \frac{1}{s_1} \leq \frac{s_n^{n-1}}{\prod_{i=1}^n s_i} .
\end{equation}
For the norm we have $\| V^{-1} \| = 1/s_1$ and $\| V \| = s_n$. For the determinant of $V$ there holds $\abs{\det V} = \prod_{i=1}^n s_i$. Hence, Ineq. \eref{eq:norm_estimate} follows from Ineq. \eref{eq:2}. To prove Ineq. \eref{eq:average_norm} recall that, since $V$ is invertible, the set $\{r \in \RR \colon \mbox{$A+rV$ is singular}\}$ is a discrete set. Thus, for almost all $r \in [-R,R]$ we may apply Ineq. \eref{eq:norm_estimate} to the matrix $A+rV$ and obtain
\[
 \bigl\| (A+rV)^{-1} \bigr\|^{s/n} \leq \frac{(\| A \| + R \| V \|)^{s(n-1)/n}}{\abs{\det (A+rV)}^{s/n}} .
\]
Inequality \eref{eq:average_norm} now follows from Lemma \ref{lemma:det}.
\end{proof}
\begin{proof}[Proof of Theorem \ref{t:a-priori}]
To avoid notation we assume $\Gamma = \ZZ$. Since $\supp \rho \subset [-R,R]$, $H_\omega$ is a bounded operator. Set $m = \| H_\Gamma \| + 1$. If $| z| \geq m$, we use $\| G_\omega (z) \| = \sup_{\lambda \in \sigma (H_\omega)} | \lambda - z |^{-1} \leq 1$ and obtain the statement of the theorem. Thus it is sufficient to consider $| z| \leq m$.
If $| x-y| \geq 4n$ Theorem 5.2 applies, since $r \leq n$. We thus only consider the case $| x-y| \leq 4n-1$. By translation we assume $x = 0$ and by symmetry $y \geq 0$. Set $\Lambda_+ = \{-1,\dots,4n\}$ and $\Lambda = \{0,\dots,4n-1\}$. Lemma \ref{lemma:fraction} gives
\[
 \Pro_{\Lambda_+} G_\omega (z) \Pro_{\Lambda_+}^* = (H_{\Lambda_+} - B_\ZZ^{\Lambda_+} - z)^{-1}
\]
where $\langle\delta_x , B_\Gamma^{\Lambda_+} \delta_y \rangle = \sum_{k \in \Gamma \setminus \Lambda_+ , \abs{k-x} = 1} \langle \delta_k , (H_{\Gamma \setminus \Lambda_+} - z)^{-1}\delta_k \rangle$ if $x=y$ and $x \in \partial \Lambda_+ = \{-1,4n\}$, and zero else.
Similarly, by another application of the Schur complement formula
\begin{eqnarray*} \fl
 \Pro_{\Lambda} (H_{\Lambda_+} - &B_\ZZ^{\Lambda_+} - z)^{-1} \Pro_{\Lambda}^* = \\ \fl
 &\Bigl(H_{\Lambda} - z - \Pro_{\Lambda} \Delta \Pro_{\partial\Lambda_+}^* \, \Bigl(\Pro_{\partial \Lambda_+}^{\Lambda_+}(H_{\Lambda_+} - B_\ZZ^{\Lambda_+})\bigl(\Pro_{\partial \Lambda_+}^{\Lambda}\bigr)^* - z \Bigr)^{-1} \, \Pro_{\partial\Lambda_+} \Delta \Pro_{\Lambda}^* \Bigr)^{-1} ,
\end{eqnarray*}
and consequently
\begin{equation} \label{eq:2xschur}
 \Pro_{\Lambda} G_\omega (z) \Pro_{\Lambda}^* = \Bigl(H_{\Lambda} - z - \Pro_{\Lambda} \Delta \Pro_{\partial\Lambda_+}^* \, (K-z)^{-1} \, \Pro_{\partial\Lambda_+} \Delta \Pro_{\Lambda}^* \Bigr)^{-1}
\end{equation}
where
\[
 K = \Pro_{\partial \Lambda_+}^{\Lambda_+}(H_{\Lambda_+} - B_\ZZ^{\Lambda_+})\bigl(\Pro_{\partial \Lambda_+}^{\Lambda_+}\bigr)^* .
\]
Note that $B_\ZZ^{\Lambda_+}$ is independent of $\omega_k$, $k \in \{-1,\dots,3n+1\}$,
and $K$ is independent of $\omega_k$, $k \in \{0,\dots,3n\}$. Thus, in matrix representation with respect to the canonical basis, the operator $K:\ell^2 (\partial \Lambda_+) \to \ell^2 (\partial \Lambda_+)$ may be decomposed as
\begin{equation*}
[K] = \pmatrix{
    \omega_{-1}u(0) & 0 \cr
    0     & \omega_{3n+1} u(n-1) \cr
   } -
\pmatrix{
    f_1 & 0 \cr
    0     & f_2 \cr
   }
\end{equation*}
where $f_1 := \sum_{k \in \ZZ \setminus \{-1\}} \omega_k u(-1-k) - \langle \delta_{-1} B_\ZZ^{\Lambda_+} \delta_{-1} \rangle$
and $f_2 := \sum_{k \in \ZZ \setminus \{3n+1\}} \omega_k u(4n-k) - \langle \delta_{4n} B_\ZZ^{\Lambda_+} \delta_{4n} \rangle$
are independent of $\omega_{-1}$ and $\omega_{3n+1}$.
Standard spectral averaging or Lemma \ref{lemma:det} gives for all $t \in (0,1)$
\begin{equation} \label{eq:averageK} 
 \mathbb{E}_{\{-1,3n+1\}} \Bigl\{\bigl\|(K-z)^{-1}\bigr\|^t \Bigr\} \leq \bigl(| u(0)|^{-t}+| u(n-1)|^{-t}\bigr) \frac{ \| \rho \|_\infty^t 2^t t^{-t}}{1-t} .
\end{equation}
Now, the operator $H_{\Lambda}$ can be decomposed as $H_{\Lambda} = A + \sum_{k=0}^{3n} \omega_k u_k$
where $A := H_{\Lambda} - \sum_{k=0}^{3n} \omega_k u_k$ is
independent of $\omega_k$, $k \in \{0,\dots,3n\}$. Let $\alpha := (\alpha_k)_{k=0}^{3n} \in [0,1]^{3n+1}$
with $\alpha_0 \not = 0$. Similarly to the proof of Lemma \ref{lemma:finitness3}, we use the substitution $\omega_0 = \alpha_0 \zeta_0$ and $\omega_i = \alpha_i \zeta_0 + \alpha_0 \zeta_i$ for $i \in \{1,\dots,3n\}$ and obtain from Eq. \eref{eq:2xschur}
\begin{eqnarray*} \fl
 E  := \mathbb{E}_{\{0,\dots,3n\}} \Bigl\{ \bigl\|  \Pro_{\Lambda} G_\omega (z) \Pro_{\Lambda}^* \bigr\|^{s/(4n)} \Bigr\} \\
 \fl \leq \| \rho \|_\infty^{3n+1} \!\!\!\!\!\!\!\!\!\!\!\! \int\limits_{[-R,R]^{3n+1}} \!\!\!\!\!\!\!\!\!  \Bigl\|  \Bigl(A + \sum_{k=0}^{3n} \omega_k u_k - z - \Pro_{\Lambda} \Delta \Pro_{\partial\Lambda_+}^*  (K-z)^{-1}  \Pro_{\partial\Lambda_+}^\Lambda \Delta \Pro_{\Lambda}^* \Bigr)^{-1} \Bigr\|^{\frac{s}{4n}} \drm \omega_0 \dots \drm\omega_{3n} \\
 \fl \leq \| \rho \|_\infty^{3n+1} \int_{[-S, S]^{3n+1}}  \Bigl\|  \Bigl(A' + \zeta_0 \sum_{k=0}^{3n} \alpha_k u_k \Bigr)^{-1} \Bigr\|^{s/(4n)} | \alpha_0|^{3n+r} \drm \zeta_0 \dots \drm\zeta_{3n}
\end{eqnarray*}
where $A' = A + \alpha_0 \sum_{k=1}^{3n} \zeta_k u_k - z - \Pro_{\Lambda} \Delta \Pro_{\partial\Lambda_+}^* \, (K-z)^{-1} \, \Pro_{\partial\Lambda_+} \Delta \Pro_{\Lambda}^*$ and $S = R (1+\max_{i \in \{1,\dots,3n\}}$ $|\alpha_i / \alpha_0 |)/| \alpha_0|$. Since $\bigcup_{i=0}^{3n} \supp u_i = \Lambda$, there exists an $\alpha \in [0,1]^{3n+1}$ such that $\sum_{k=0}^{3n} \alpha_k u_k$ is invertible on $\ell^2 (\Lambda)$, compare the proof of Lemma \ref{lemma:finitness3} and Figure \ref{fig:volume}. Thus we may apply Lemma \ref{lemma:averagenorm} and obtain
\begin{equation} \label{eq:Eint} \fl
 E \leq \| \rho \|_\infty^{3n+1} \int_{[-S,S]^{3n}} \frac{2 s^{-s} S^{1-s}}{1-s} \frac{\Bigl(\| A' \| + S \| \sum_{k=0}^{3n} \alpha_k u_k \| \Bigr)^{\frac{s(4n-1)}{4n}}}{\abs{\det \bigl( \sum_{k=0}^{3n} \alpha_k u_k \bigr)}^{s/(4n)}} \drm \zeta_1 \dots \drm \zeta_{3n}
\end{equation}
Using $\zeta_k \in [-S,S]$ for $k \in \{1,\dots,3n\}$, $\omega_k \in [-R,R]$ for $k \in \ZZ \setminus \{0,\dots,3n\}$ and $\alpha_k \in [0,1]$ for $k \in \{0,\dots,3n\}$, the norm of $A'$ can be estimated as
\begin{eqnarray} \fl
 \| A' \| = \Bigl\| H_{\Lambda} - \sum_{k=0}^{3n} \omega_k u_k + \alpha_0 \sum_{k=1}^{3n} \zeta_k u_k - z - \Pro_{\Lambda} \Delta \Pro_{\partial\Lambda_+}^* \, (K-z)^{-1} \, \Pro_{\partial\Lambda_+} \Delta \Pro_{\Lambda}^* \Bigr\| \nonumber \\
  \leq 2 + (n-1) R \, \| u\|_{\infty} + 3S n \| u \|_{\infty} + m + 4  \bigl\|(K-z)^{-1} \bigr\| \, .\label{eq:A'}
\end{eqnarray}
All terms in the sum \eref{eq:A'} are independent of $\zeta_k$, $k \in \{0,\dots,3n\}$.
Using $(\sum | a_i|)^t \leq \sum | a_i|^t$ for $t<1$ we see from Ineq.~\eref{eq:Eint} and \eref{eq:A'} that there are constants $C_1$ and $C_2$ such that $E \leq C_1 + C_2\|(K-z)^{-1} \|^{s(4n-1)/(4n)}$.
If we average over $\omega_{-1}$ and $\omega_{3n+1}$, Ineq.~\eref{eq:averageK} gives the desired result.
\end{proof}
%
%
%
%
%
%
%
%
\section{Localization} \label{sec:loc}
In this section we discuss exponential localization for the discrete alloy-type model. In particular,
\begin{enumerate}[$\bullet$]
 \item we prove Theorem \ref{theorem:loc} on exponential localization in the one-dimensional discrete alloy-type model for large disorder, see the end of the section.
 \item we establish a criterion for exponential localization using fractional moment bounds \emph{only}, see Theorem \ref{thm:fmb-localization}.
\item we discuss why previous papers on fractional moment bounds do not imply immediately spectral localization for our model.
\end{enumerate}
The existing proofs of localization via the fractional moment method use either
the Simon Wolff criterion, see e.\,g. \cite{SimonW1986, AizenmanM1993,AizenmanSFH2001}, or the RAGE-Theorem, see e.\,g.
\cite{Aizenman1994,Graf1994}, or eigenfunction correlators, see e.\,g.
\cite{AizenmanENSS2006}. Neither dynamical nor spectral localization can be directly
inferred from the behavior of the Green's function using the existent methods.
The reason is that the random variables $V_\omega (x)$, $x \in \ZZ$, are not
independent, while the dependence of $H_\omega$ on the i.\,i.\,d. random
variables $\omega_k$, $k \in \ZZ$, is not monotone.
\par
However, for the discrete alloy-type model it is possible to show localization using the multiscale analysis. The two
ingredients of the multiscale analysis are the initial length scale estimate and
the Wegner estimate, compare assumptions (P1) and (P2) of \cite{DreifusK1989}.
The initial length scale estimate is implied by the exponential decay of an
averaged fractional power of Green's function, i.\,e. Theorem
\ref{theorem:exp} and \ref{theorem:exp2}, using Chebyshev's inequality. A Wegner
estimate for the models considered here was established in \cite{Veselic2010b}. Thus a variant of the multiscale analysis of
\cite{DreifusK1989} yields pure point spectrum with exponential decaying
eigenfunctions for almost all configurations of the randomness. We say a
variant, since in our case the potential values are independent only for lattice
sites having a minimal distance. It has been implemented in detail in the paper
\cite{GerminetK2001} for random Schr\"odinger operators in the continuum, and holds similarly for discrete models.
\par
In this section we conclude localization from bounds on averaged fractional powers of Green's function \emph{without using the multiscale analysis}.
Roughly speaking, we skip over the induction step of the multiscale analysis and directly compute the ``typical output''
of the multiscale analysis, i.\,e. the hypothesis of Theorem 2.3 in \cite{DreifusK1989}.
Then we conclude localization using existent methods.
\par
For $L > 0$ and $x \in \ZZ$ we denote by $\Lambda_{x,L} = [ x-L , \dots ,
x+L ] \cap \ZZ$ the cube of side length $2L+1$. Let further $m > 0$, $L \in \NN$ and $E \in \RR$. A cube $\Lambda_{x,L}$ is called \emph{$(m,E)$-regular} (for a fixed
potential), if $E \not \in \sigma (H_{\Lambda_{x,L}})$ and
\[
 \sup_{w \in \partial \Lambda_{x,L}} | G_{\Lambda_{x,L}} (E ; x,w) |
\leq \euler^{-m L} .
\]
Otherwise we say that $\Lambda_{x,L}$ is \emph{$(m , E)$-singular}. The next
Proposition now states that certain bounds on averaged fractional moments of Green's function imply the hypothesis of Theorem 2.3 in \cite{DreifusK1989} (without applying the induction step of the multiscale analysis).
\begin{proposition} \label{prop:replace-msa}
Let $\Theta \subset [0,n-1] \cap \ZZ$ for some $n \in \NN$, $I \subset \RR$
be a bounded interval and $s \in (0,1)$. Assume the following two statements:
\begin{enumerate}[(i)]
\item  There are constants $C,\mu \in (0,\infty)$, $L_0 \in \NN$ and $N \in \NN$
such that $\EE \bigl\{ | G_{\Lambda_{k,L}} (E;x,y) |^{s/N} \bigr\} \leq C \euler^{-\mu | x-y |}$ for all $k \in \ZZ$, $L \in \NN$, $x,y \in \Lambda_{k,L}$ with $| x-y | \geq L_0$, and all $E \in I$.
\item There are constants $C' \in (0,\infty)$ and $N' \in \NN$ such that $\EE \bigl\{ | G_{\Lambda_{k,L}} (E+\i \epsilon ;x,x) |^{s/N'} \bigr\} \leq C'$ for all $k \in \ZZ$, $L \in \NN$, $x \in \Lambda_{k,L}$, $E \in I$ and $\epsilon>0$ .
\end{enumerate}
Then we have for all $L \geq \max\{8 \ln (2)/\mu , L_0\}$ and all $x,y \in \ZZ$ with $| x-y| \geq 2L+n$ that
\begin{eqnarray*} 
 \PP \{\forall \, E \in I \mbox{ either $\Lambda_{x,L}$ or $\Lambda_{y,L}$ is
$(\mu/8,E)$-regular} \}  \geq  1- K ,
\end{eqnarray*}
where $K = 8 \bigl( C | I |  + C_{\rm W}(2L+1)^2 \bigr) \euler^{-\mu sL /(8N')}$ and $C_{\rm W}=4C'/\pi$.
\end{proposition}
\begin{proof}
 Fix $L \in \NN$ with $L \geq \max\{8 \ln (2)/\mu , L_0\}$ and $x,y \in \ZZ$ such that
$| x-y | \geq 2L+n$. For $\omega \in \Omega$ and $k \in
\{x,y\}$ we define the sets
\begin{eqnarray}
 \Delta_\omega^k &:= \{E \in I : \sup_{w \in \partial \Lambda_{k,L}} |
G_{\Lambda_{k,L}} (E ; k,w) | > \euler^{-\mu L /8}\}, \nonumber \\
\tilde \Delta_\omega^k &:= \{E \in I : \sup_{w \in \partial \Lambda_{k,L}}
| G_{\Lambda_{k,L}} (E ; k,w) | > \euler^{-\mu L/4 }\}, \mbox{ and} \nonumber \\
 \tilde B_k &:= \{\omega \in \Omega : \mathcal{L} \{\tilde \Delta_\omega^k\} >
 \euler^{-5\mu L /8} \} . \label{eq:deltatilde}
\end{eqnarray}
Here $\mathcal{L}$ denotes the Lebesgue measure. For $\omega \in \tilde B_k$ we have
\begin{eqnarray*}
\sum_{w \in \partial \Lambda_{k,L}} \!\!\!\!  \int_I | G_{\Lambda_{k,L}} (E ; k,w) |^{s/N} \drm E
& \ge \int_I \sup_{w \in \partial \Lambda_{k,L}} | G_{\Lambda_{k,L}} (E ; k,w) |^{s/N} \drm E \\
& > \euler^{-5 \mu L / 8} \euler^{-\mu L s/ (4N)} > \euler^{-7 \mu L / 8}.
\end{eqnarray*}
By assumption (i) this implies for $k \in \{x,y\}$
\begin{equation*}
 \PP \{\tilde B_k\} <  2 | I | C \euler^{-\mu L /8} .
\end{equation*}
For $k \in \{x,y\}$ we denote by $\sigma (H_{\Lambda_{k,L}}) =
\{E_{\omega,k}^i\}_{i=1}^{2L+1}$ the spectrum of
$H_{\Lambda_{k,L}}$. We claim that for $k \in \{x,y\}$,
\begin{equation} \label{eq:claim}
 \omega \in \Omega \setminus \tilde B_k \quad \Rightarrow
 \quad \Delta_\omega^k \subset \bigcup_{i=1}^{2L+1}
 \bigl[E_{\omega,k}^i-\delta , E_{\omega,k}^i + \delta \bigr] =:
 I_{\omega,k}(\delta),
\end{equation}
where $\delta = 2\euler^{-\mu L / 8}$. Indeed, suppose that $E\in \Delta_\omega^k$ and $\dist\big(E,\sigma
(H_{\Lambda_{k,L}})\big)>\delta$. Then there exists $w \in \partial \Lambda_{k,L}$ such
that $| G_{\Lambda_{k,L}} (E;k,w) | > \euler^{-\mu L / 8}$.
For any $E'$ with $| E-E'| \le 2\euler^{-5\mu L / 8}$ we have
$\delta -| E-E'|\ge \euler^{-\mu L / 8}\ge 2\euler^{-3\mu L / 8}  $ since $L \geq 8 \ln (2) / \mu$.
Moreover, the first resolvent identity and the estimate $\| (H-E)^{-1} \| \leq \dist
(E,\sigma (H))^{-1}$ for selfadjoint $H$ and $E \in \CC\setminus \sigma (H)$ implies
\begin{eqnarray*} \fl
| G_{\Lambda_{k,L}} (E ; k,w) - G_{\Lambda_{k,L}} (E' ; k,w)| \leq \ | E-E'| \cdot \| G_{\Lambda_k} (E)\| \cdot \| G_{\Lambda_k} (E') \| \leq \frac{1}{2} \euler^{-\mu L / 8} ,
\end{eqnarray*}
 and hence
\[
 | G_{\Lambda_{k,L}} (E' ; k,w)|
\ > \ \frac{\euler^{-\mu L / 8}}{2} \geq \euler^{-\mu L / 4}
\]
for $L \geq 8 \ln (2) / \mu$.
We infer that $[E-2\euler^{-5\mu L / 8},E+2\euler^{-5\mu L / 8}]\cap I \subset  \tilde \Delta_\omega^k$ and conclude
$\mathcal{L} \{\tilde \Delta_\omega^k\} \ge 2\euler^{-5\mu L / 8}$.
This is however impossible if
$\omega\in \Omega \setminus \tilde B_k$ by \eref{eq:deltatilde},
hence the claim \eref{eq:claim} follows.
\par
In the following step we use assumption (ii) to deduce a  Wegner-type estimate.
We denote by $P_{[a,b]} (H_{\Lambda_{x,L}})$ the spectral projection corresponding to the interval $[a,b] \subset \RR$ and the operator $H_{\Lambda_{x,L}}$. Since we have for any $a,b \in \RR$ with $a<b$, any $\lambda \in \RR$ and $0<\epsilon \leq b-a$
\[
 \arctan \left( \frac{\lambda - a}{\epsilon} \right) - \arctan \left( \frac{\lambda - b}{\epsilon} \right) \geq \frac{\pi}{4} \ \chi_{[a,b]}(\lambda) ,
\]
one obtains an inequality version of Stones formula:
\[
 \langle \delta_x , P_{[a,b]} (H_{\Lambda_{x,L}}) \delta_x \rangle
\leq \frac{4}{\pi} \int_{[a,b]} \im \left\{ G_{\Lambda_{x,L}} (E+ \i \epsilon ; x,x) \right\} \drm E 
\]
for all $\epsilon \in (0, b-a]$. Using triangle inequality, $| \im z| \leq | z|$ for $z \in \CC$, Fubini's theorem,
$| G_{\Lambda_{x,L}} (E+\i \epsilon ; x,x) |^{1-s/N'} \leq \dist (\sigma(H_{\Lambda_{x,L}}) , E+i \epsilon)^{s/N'-1} \leq \epsilon^{s/N'-1}$
and assumption (ii) we obtain for $[a,b]\subset I$ and all $\epsilon \in (0,b-a]$
\begin{eqnarray*} \fl
\EE \bigl\{ \Tr P_{[a,b]}(H_{\Lambda_{x,L}}) \bigr\} & \leq \EE \Bigl\{ \sum_{x \in \Lambda_{x,L}} \frac{4}{\pi} \int_{[a,b]} \im \left\{ G_{\Lambda_{x,L}} (E+\i \epsilon ; x,x) \right\} \drm E  \Bigr\} \\
&  \leq  \frac{4\epsilon^{s/N'-1}}{\pi}  \sum_{x \in \Lambda_{x,L}} \int_{[a,b]} \EE \Bigl\{ \bigl|  G_{\Lambda_{x,L}} (E+\i \epsilon ; x,x)  \bigr|^{s/N'} \Bigr\} \drm E   \\
& \leq 4\pi^{-1}\epsilon^{s/N'-1}  (2L+1) \, | b-a | C' .
\end{eqnarray*}
We minimize the right hand side by choosing $\epsilon = b-a$ and obtain the Wegner estimate
\begin{eqnarray} 
\EE \bigl\{ \Tr P_{[a,b]}(H_{\Lambda_{x,L}}) \bigr\} &\leq \frac{4C'}{\pi} | b-a |^{s/N'}  (2L+1) \nonumber \\
& =: C_{\rm W} | b-a |^{s/N'}  (2L+1) .\label{eq:wegner}
\end{eqnarray}
Now we want to estimate the probability of the event $B_{\rm res} := \{\omega \in \Omega :  I \cap I_{\omega,x}(\delta) \cap I_{\omega,y}(\delta) \not = \emptyset \}$
that there are ``resonant'' energies for the two box Hamiltonians $H_{\Lambda_{x,L}}$ and $H_{\Lambda_{y,L}}$.
For this purpose we denote by $\Lambda_{x,L}^+$ the set of all lattice sites $k \in \ZZ$
whose coupling constant $\omega_k$ influences the potential in $\Lambda_{x,L}$,
i.\,e. $\Lambda_{x,L}^+ = \cup_{x \in \Lambda_{x,L}} \{k \in \ZZ : u(x-k) \not = 0)\}$.
Notice that the expectation in Ineq. \eref{eq:wegner} may therefore be replaced by $\EE_{\Lambda_{x,L}^+}$.
Moreover, since $| x-y | \geq 2L + n$ and $\Theta \subset [0,n-1] \cap \ZZ$,
the operator $H_{\Lambda_{y,L}}$ and hence the interval $I_{\omega , y} (\delta)$ is independent of $\omega_k$, $k \in \Lambda_{x,L}^+$.
We use the product structure of the measure $\PP$, Chebyshev's inequality, and estimate \eref{eq:wegner} to obtain
\begin{eqnarray} \fl
\PP_{\Lambda_{x,L}^+}(B_{\rm res}) &= \sum_{i=1}^{2L+1} \PP_{\Lambda_{x,L}^+} \bigl\{ \omega \in \Omega : \Tr \bigl( P_{I \cap [E_{\omega,y}^i-2\delta , E_{\omega,y}^i + 2\delta ]} (H_{\Lambda_{x,L}}) \bigr) \geq 1 \bigr\} \nonumber  \\ \fl
& \leq \sum_{i=1}^{2L+1} \EE_{\Lambda_{x,L}^+} \bigl\{ \Tr \bigl( P_{I \cap [E_{\omega,y}^i-2\delta , E_{\omega,y}^i + 2\delta ]}  (H_{\Lambda_{x,L}}) \bigr) \bigr\} \nonumber \\ \fl
&\leq (2L+1)^2  C_{\rm W} (4\delta)^{s/N'} . \label{eq:wegner-application}
\end{eqnarray}
Consider now an $\omega \not \in \tilde B_x \cup \tilde B_y$. Recall that \eref{eq:claim} tell us that $\Delta_\omega^x \subset  I_{\omega,x}(\delta)$
and $\Delta_\omega^y \subset  I_{\omega,y}(\delta)$. If additionally $\omega \not \in B_{\rm  res}$ then no $E \in I$ can be in
$\Delta_\omega^x $ and $\Delta_\omega^y$ simultaneously. Hence for each $E \in I$ either  $\Lambda_{x,L}$ or $\Lambda_{y,L}$
is $(\mu/8,E)$-regular. A contraposition gives us
\begin{eqnarray*} \fl
\PP \bigl\{\mbox{$\exists \, E \in I$, $\Lambda_{x,L}$ and $\Lambda_{y,L}$ is $(\mu/8,E)$-singular} \bigr\}
& \le \PP (\tilde B_x) +\PP (\tilde B_y ) + \PP(B_{\rm res}) \\
&\le 4 | I | C \euler^{-\frac{\mu L}{8}} + (2L+1)^2  C_{\rm W} (4\delta)^{\frac{s}{N'}} ,
\end{eqnarray*}
which proves the statement of the proposition.
\end{proof}
\begin{remark}
In the proof of Proposition \ref{prop:replace-msa} assumption (ii) is only used to get a Wegner estimate, see Ineq. \eref{eq:wegner}. Hence, if we know that a Wegner estimate holds for some other reason, we can drop assumption (ii) and skip a part of the proof of Proposition \ref{prop:replace-msa}. In the situation where $\rho \in W^{1,1} (\RR)$, a Wegner estimate for our model was established in \cite{Veselic2010b}.
\end{remark}
\begin{remark} \label{remark:assii}
Note that the conclusions of Proposition \ref{prop:replace-msa} tells us that the probabilities
of $\{\forall \, E \in I \mbox{ either $\Lambda_{x,L}$ or $\Lambda_{y,L}$ is $(\mu/8,E)$-regular} \}$
tend to one exponentially fast as $L$ tends to infinity. In particular, for any $p>0$ there is some $\tilde L \in \NN$
such that for all $L \ge \tilde L$:
\[
 \PP \{\forall \, E \in I \mbox{ either $\Lambda_{x,L}$ or $\Lambda_{y,L}$ is $(m,E)$-regular} \}\ge 1- L^{-2p}.
\]
\end{remark}
To conclude exponential localization from the estimate provided in Proposition
\ref{prop:replace-msa} we will use Theorem 2.3 in \cite{DreifusK1989} in the $1d$ situation. Since the potential values at different lattice sites are not necessarily independent in our model we need a slight extension which can be proven with the same arguments as the original result.
\begin{theorem}[\cite{DreifusK1989}] \label{thm:vDK-2.3}
Let $I\subset \RR$ be an interval and let $p>1$, $ L_0>1$, $\alpha \in (1,2p)$, $m>0$, $n \in \NN$.
Assume that $\Theta \subset     [0,n-1]\cap \ZZ$.
Set $L_{k} =L_{k-1}^\alpha$, for $k \in \NN$. Suppose that for any $k \in \NN_0$
\[
 \PP \{\forall \, E \in I \mbox{ either $\Lambda_{x,L_k}$ or $\Lambda_{y,L_k}$ is $(m,E)$-regular} \}\ge 1- L_k^{-2p}
\]
for any $x,y \in \ZZ$ with $| x-y| > 2 L_k + n$.

Then for almost all $\omega \in \Omega$ there is no continuous spectrum of $H_\omega$ in $I$ and
the eigenfunctions of all eigenvalues of $H_\omega$ in $I$ decay exponentially at infinity with mass $m$.
\end{theorem}
Putting together Proposition \ref{prop:replace-msa} and
Theorem \ref{thm:vDK-2.3} we obtain the statement that exponential
decay of fractional moments of the Green's function implies exponential localization.
\begin{theorem}\label{thm:fmb-localization}
Let $n \in \NN$, $\Theta \subset [0,n-1]\cap \ZZ$, $s \in (0,1)$,
$C,\mu \in (0,\infty)$,  $N \in \NN$,
and $I \subset \RR$ be a bounded interval. Assume that
$\EE \bigl\{ | G_{\Lambda_{k,L}} (E;x,y) |^{s/N} \bigr\} \leq C \euler^{-\mu | x-y |}$
for all $k \in \ZZ^d$, $L \in \NN$, $x,y \in \Lambda_{k,L}$ and all $E \in I$.
Then for almost all $\omega \in \Omega$ there is no continuous spectrum of $H_\omega$ in $I$ and the eigenfunctions of all eigenvalues of $H_\omega$ in $I$ decay exponentially at infinity.
\end{theorem}
We now verify the hypotheses of Propositions \ref{prop:replace-msa}.
The next lemma shows that for finite box restrictions of our model the resolvent is well defined at almost all energies.

\begin{lemma} \label{lemma:nullset}
 Let $\min \Theta = 0$ and $\max \Theta = n-1$ for some $n \in \NN$. Let further
$\mathcal{F}$ denote the set of all finite connected subsets of $\ZZ$. Then, for
each $E \in \RR$ the set
\begin{equation*} \label{eq:union_null}
\bigcup_{\Gamma \in \mathcal{F}} \Bigl\{ \omega \in \Omega \mid E \in \sigma
(H_\Gamma) \Bigr\}
\end{equation*}
has $\PP$-measure zero.
\end{lemma}
\begin{proof}
Let $\Lambda = \{0,1,\dots,m\}$ for some $m\in \NN$. The potential $V_\Lambda$
depends on the random variables $\omega_k$, $k \in \{-n+1,-n+2,\dots,m\}
= \Lambda^+$. For computing the probability $P = \PP \bigl\{
\det(H_{\Lambda} - E)=0 \bigr\}$ we have the estimate
\[
P  \leq \| \rho \|_\infty^{m+n}
\int_{\RR^{m+n}} \chi_{\bigl\{ \omega \in \Omega  \mid \det (H_\Lambda - E) = 0
\bigr\} } (\omega) \drm \omega
\]
where $\omega = (\omega_k)_{k \in \Lambda^+}$ and $\drm \omega = \prod_{k \in
\Lambda^+} \drm \omega_k$.
We introduce the linear transformation $\omega_k = r_k$ for $k \in
\{-n+1,\dots,-1\}$, $\omega_0 = \alpha_0 r_0$ and $\omega_k = \alpha_k r_0 +
\alpha_0 r_k$ for $k \in \{1,2,\dots,m\}$ with $\alpha_0 = 1$.
With this transformation we can rewrite the potential as
\begin{eqnarray*} \fl
 V_{\Lambda} &= \sum_{k=-n+1}^{-1} \omega_k V_k + \sum_{k=0}^{m}  \omega_{k}
V_k
             =\sum_{k=-n+1}^{-1} r_k V_k  + \sum_{k=1}^{m} \alpha_0 r_{k} V_k +
r_{0} \sum_{k=0}^{m} \alpha_k V_k ,
\end{eqnarray*}
where $V_k$, $k=-n+1,\dots,m$, are diagonal matrices with diagonal elements
$u(i-k)$, $i=0,\dots,m$. We introduce the notation
\[
A = -\Delta_\Lambda + \sum_{k=-n+1}^{-1} r_k V_k  + \sum_{k=1}^{m}
\alpha_0 r_{k} V_k - E \quad \mbox{and} \quad V = \sum_{k=0}^{m}
\alpha_k V_k ,
\]
transform the integral and obtain
\begin{equation} \label{eq:ptransform}
P  \leq \| \rho \|_\infty^{m+n}
\int_{\RR^{m+n}} \chi_{\bigl\{ r \in \Omega  \mid \det (A + r_0 V) = 0 \bigr\} }
(r) \drm r ,
\end{equation}
where $r = (r_k)_{k \in \Lambda^+}$ and $\drm r = \prod_{k \in \Lambda^+} \drm
r_k$.
For the determinant of $V$ we have
\[
 \det (V) = \prod_{i=0}^{m} \left( \sum_{k=0}^{m} \alpha_k u(i-k) \right) =
\prod_{i=0}^{m} \langle \alpha , u_i \rangle ,
\]
where we have used the notation $u_i = (u(i-k))_{k=0}^{m}$ and $\alpha =
(\alpha_k)_{k=0}^m$. From $u(0) \not = 0$ we conclude that $u_i \not \equiv 0$
for all $i \in \{0,\dots,m\}$. We choose the vector $\alpha$ with $\alpha_0 = 1$
such that $\langle \alpha , u_i \rangle \not = 0$ for all $i \in \{0,\dots,m\}$.
With this choice of $\alpha$ the matrix $V$ is invertible. Thus, for any fixed
collection of $r_k$, $k \in \Lambda^+ \setminus \{0\}$, the mapping $\RR \ni r_0
\mapsto \det (A + r_0 V)$ is a polynomial of order $m+1$. Therefore, the set
\[
 \{r_{0} \in \RR \mid \det (A + r_{0} V)=0\}
\]
is a finite set. From this property, Fubini's theorem and Ineq.
\eref{eq:ptransform} we obtain $P = 0$. By translation this gives for arbitrary
finite connected set $\Lambda \subset \ZZ$ and each $E \in \RR$ that $\PP
\bigl\{E \in \sigma (H_\Lambda) \bigr\} = 0$. Since the union
\eref{eq:union_null} is countable, we obtain the statement of the lemma.
\end{proof}
%
%
%
%
\begin{remark} \label{remark:real}
 Lemma \ref{lemma:nullset} allows us to obtain the results of Theorem
\ref{theorem:exp}, \ref{theorem:exp2} and \ref{t:a-priori} also for real energies $z$ in the case
where $\Gamma$ is a finite set. For sets of measure zero, the integrand may not
be defined. However, for the bounds on the expectation value this is irrelevant.
\end{remark}
\begin{proof}[Proof of Theorem \ref{theorem:loc}]
Since $\RR = \cup_{M \in \ZZ} [M,M+1]$ is a countable union of compact intervals, it is sufficient to
show that the assumptions of Theorem \ref{thm:vDK-2.3} hold for each interval $[M,M+1]$ individually.
The assumptions of Proposition \ref{prop:replace-msa} are fulfilled by Theorem \ref{thm:result2}, Theorem \ref{t:a-priori} and Remark \ref{remark:real}. Hence we obtain the statement of the theorem if $\| \rho \|_\infty$ is sufficiently small.
\end{proof}
%
%
%
%
%
%
%
%
%
%
\appendix
\section{Reduction to the monotone case} \label{sec:monotone}
This appendix concerns the last statement of Remark \ref{remark:localization}. First we discuss a criterion which ensures that an appropriate one-parameter family of positive
potentials can be extracted from the random potential $V_\omega$.
\begin{lemma} \label{l:equivalence}
Let $u= \sum_{k=0}^{n-1}u(k) \delta_k\colon \ZZ\to \RR$.
Then the following statements are equivalent.
\begin{enumerate}[(A)]
\item
There exists an $N \in \NN$ and real $\lambda_0, \dots, \lambda_N$
such that $w := u \ast \lambda := \lambda_0 u_0 + \dots + \lambda_N u_N$ is a non-negative function and
$w(0)>0 $, $w(N+n-1)>0 $ hold.
\item
There exists an $M \in \NN$ and real $\gamma_0, \dots, \gamma_M$
such that $v := u \ast \gamma := \gamma_0 u_0 + \dots + \gamma_M u_M$ is a non-negative function and
$\supp v = \{0, \dots, M+n-1\}$ holds.
\item
The polynomial $\CC \ni z \mapsto p_u(z) := \sum_{k=0}^{n-1}u(k) z^k$ has no roots in $[0,\infty)$.
\end{enumerate}
\end{lemma}
Note, if $u(0) \not = 0$ and $u(n-1) \not = 0$, then $\{0, \dots, M+n-1\}$ is the union of the supports of $u_0, \dots, u_M$.
If (A) or (B) hold
we may assume that $| \lambda_0| $, respectively $ |\gamma_0| $,  equals one.
\begin{proof}
If (A) holds, one may choose $v (x) = \sum_{j=0}^{N+n-2} w(x-j)$ to conclude (B). Thus it is sufficient to show (B)$\Leftrightarrow$(C). Using Fourier transform and the identity theorem for holomorphic functions one sees that (B) is equivalent to
\begin{enumerate}[(D)]
\item There exists an $M \in \mathbb{N}$ and real $\gamma_0,\dots,\gamma_M$ such that all coefficients of the polynomial $p_u(z) \cdot \sum_{j=0}^M \gamma_j z^j$ are strictly positive.
\end{enumerate}
If (D) holds, $p_u(x) \cdot \sum_{j=0}^M \gamma_j x^j$ is strictly positive for $x \in [0,\infty)$. Thus its divisor $p_u$ has no root in $[0,\infty)$ and one concludes (C). Assuming (C), one infers from Corollary 2.7 of \cite{MotzkinS1969} that there exists a polynomial $p$ such that $p_u \cdot p$ has strictly positive coefficients. Choosing $M = \deg (p)$ and $\gamma_0,\dots,\gamma_M$ to be the coefficients of $p$ leads to (D).
\end{proof}
If the random potential $V_\omega$ contains a positive building block
$w$ as in (A) of the previous lemma, one
obtains Theorems \ref{thm:result2} with \cite{AizenmanENSS2006}, as we outline now.
The crucial tool is Proposition 3.2 of \cite{AizenmanENSS2006}. Here are two direct consequences of the latter:
\begin{lemma}\label{l:AENSS}
Let $H$ be bounded, selfadjoint on $\ell^2(\ZZ)$,
$\phi,\psi\colon\ZZ \to [0,\infty)$ bounded, $z \in \CC$ with $\im z >0$, and $t,S\in (0,\infty)$.
Then there is a universal constant $C_W\in (0,\infty)$
such that for all $x, y\in \ZZ$
\begin{enumerate}[(i)]
\item
$ \displaystyle
\sqrt{\phi(x)\psi(y)} \,\mathcal{L} \bigl\{ v_1,v_2 \in [-S,S] : | \langle\delta_x , (H+z - v_1 \phi - v_2 \psi)^{-1}\delta_y \rangle | >t  \bigr \}
 \leq 4 C_W \frac{S}{t} \,
$\\[0,5em]
where $\mathcal{L}$ denotes Lebesgue measure.
\item If $\phi(x)\psi(y)\neq0$ and $s \in (0,1)$:
\begin{equation*} \fl
\int_{[-S,S]^2}
|\langle\delta_x , (H+z - v_1 \phi - v_2 \psi)^{-1}\delta_y \rangle |^s \drm v_1 \drm  v_2
 \leq
\frac{4}{1-s} \left(\frac{C_W}{\sqrt{\phi(x)\psi(y)}}\right)^s S^{2-s}\, .
\end{equation*}
\end{enumerate}
\end{lemma}
To obtain statement (ii) from (i) use the layer cake representation
\[
\int_{[-S,S]^2}  | f(v_1,v_2)|^s \drm v_1 \drm  v_2
= \int_0^\infty  \mathcal{L} \{| v_1|,| v_2| \le S : | f(v_1,v_2)|^s >t\} \drm t
\]
and decompose the integration domain into $[0, \kappa]$ and $[\kappa, \infty)$
where $\kappa = \big (C_w / S \sqrt{\phi(x)\psi(y)}\big )^s$.

\begin{proposition} \label{p:bounded}
Let $\Gamma \subset \ZZ$ be connected, $\Theta \subset \ZZ$
with $\min \Theta = 0$ and $\max \Theta = n-1$ for some $n \in \NN$.
Assume that  $u$ satisfies condition (A) in Lemma \ref{l:equivalence} and that $\supp \rho$ is compact. Set $\Lambda_x = \{x,\dots,x+N\}$ and $\Lambda_j=\{j-n+1-N,\dots,j-n+1\}$.
Then we have for all $x,j \in \Gamma$ with $| j-x| \geq 2(N+n) -1$
and all $z \in \CC$ with $\im z>0$
\[
 \mathbb{E}_\Lambda \bigl\{ | G_\Gamma (z;x,j) |^s  \bigr\} \leq C
\]
where $C$ is defined in Eq. \eref{e:K} and $\Lambda= \Lambda_x \cup \Lambda_j$.
\end{proposition}
\begin{proof}
Without loss of generality we assume $j-x \geq 2(N+n) -1$ and $\lambda_0=1$.
By assumption $\Gamma \supset \{x,x+1,\dots,j\}$. Note that the operator
$A' := H_\Gamma - z -\sum_{k \in \Lambda_x} \omega_k u_k - \sum_{k \in \Lambda_j} \omega_k u_k$ is independent of $\omega_k$, $k \in \Lambda$.
To estimate the expectation
\begin{eqnarray*} \fl
 E &:= \mathbb{E}_{\Lambda} \Bigl\{ \bigl| G_\Gamma (z;x,j) \bigr|^s  \Bigr\} \\ \fl
&= \int_{[-R,R]^{|\Lambda|}} \Bigl| \Bigl\langle\delta_x , \Bigl(A' + \sum_{k \in \Lambda_x} \omega_k u_k + \sum_{k \in \Lambda_j} \omega_k u_k \Bigr)^{-1} \delta_j\Bigr\rangle \Bigr|^s \prod_{k \in \Lambda} \rho(\omega_k) \drm \omega_k
\end{eqnarray*}
we use the substitutions
\[ \fl
 \pmatrix{
  \omega_x \cr
  \omega_{x+1} \cr
  \vdots \cr
  \vdots \cr
  \omega_{x+N} \cr
 }
=
 T
\pmatrix{
  \zeta_x \cr
  \zeta_{x+1} \cr
  \vdots \cr
  \vdots \cr
  \zeta_{x+N} \cr
}
\quad \mbox{and} \quad
 \pmatrix{
  \omega_{j-n+1-N} \cr
  \omega_{j-n+2-N} \cr
  \vdots \cr
  \vdots \cr
  \omega_{j-n+1} \cr
 }
=
 T
 \pmatrix{
  \zeta_{j-n+1-N} \cr
  \zeta_{j-n+2-N} \cr
  \vdots \cr
  \vdots \cr
  \zeta_{j-n+1} \cr
 }
\]
where the matrix $T$ is the same as in Lemma \ref{lemma:detgen}
with $\alpha_k$ replaced by $\lambda_k$, $k =0, \dots, N$.
This gives the bound
\[ \fl
  E \leq \| \rho \|_\infty^{|\Lambda|} \!\!\!\!\!\!\!\!\! \int\limits_{[-S,S]^{|\Lambda|}}\!\!\!\!\!\!
\Bigl| \Bigl\langle \delta_x ,
    \Bigl(A + \zeta_x \! \sum_{k \in \Lambda_x} \! \lambda_{k-x} u_k +  \zeta_{j-n+1-N} \!\sum_{k \in \Lambda_j} \! \lambda_{k-(j-n+1-N)} u_k \Bigr)^{-1}
\delta_j \Bigr\rangle \Bigr|^s
\drm \zeta_\Lambda .
\]
where $\drm \zeta_\Lambda = \prod_{k \in \Lambda} \drm \zeta_k$,
$S = R (1+\max_{i\in \{1,\dots,N\}} |\lambda_i|)$, and
\[
A := A' + \sum_{k \in \Lambda_x \setminus\{x\}}  \zeta_k u_k + \sum_{k \in \Lambda_j \setminus\{j-n+1-N\}}  \zeta_k u_k
\]
 is independent of $\zeta_x$ and $\zeta_{j-n+1-N}$.
By assumption the functions
 $\phi:= \sum_{k \in \Lambda_x} \lambda_{k-x} u_k$ and $\psi:= \sum_{k \in \Lambda_j} \lambda_{k-(j-n+1-N)} u_k$ are bounded and  non-negative, with $\phi(x)=  u(0) > 0$ and $\psi(j) = \lambda_N u(n-1) > 0$.
Using Lemma \ref{l:AENSS} we obtain
\begin{eqnarray*}
 E' &:= \int_{[-S,S]^{2}} \Bigl| \Bigl\langle \delta_x , \bigl(A + \zeta_x \phi + \zeta_{j-n+1-N} \psi \bigr)^{-1} \delta_j \Bigr\rangle \Bigr|^s \drm \zeta_x \drm \zeta_{j-n+1-N} \\[1ex]
&\leq
\frac{4}{1-s} \left(\frac{C_W}{\sqrt{\phi(x)\psi(j)}}\right)^s S^{2-s} .
\end{eqnarray*}
Thus the original integral is estimated by
\begin{eqnarray}
  E &\leq
\| \rho \|_\infty^{|\Lambda|}  (2S)^{|\Lambda|-2}
\frac{4}{1-s} \left(\frac{C_W}{\sqrt{\phi(x)\psi(j)}}\right)^s S^{2-s} \nonumber
 \\[1ex]
 &=
\frac{4}{1-s} \left(\frac{C_W}{\sqrt{u(0) \lambda_N u(n-1) }}\right)^s
\big(2S\| \rho \|_\infty \big)^{2(N+1)} \frac{1}{S^{s}} =: C . \label{e:K}
\end{eqnarray}
\end{proof}
The last proposition and a formula analogous to \eref{e:geometric-resolvent} now give for $j=x + 2(N+n)-1$ and
$x + 2(N+n) \le y$
\begin{eqnarray*} \fl
\EE_\Lambda \Bigl\{ \bigl| G_{\Gamma_0} (z;x,y) \bigr|^s  \Bigr\} &= \EE_\Lambda \Bigl\{ \bigl| G_{\Gamma_0} (z;x,x + 2(N+n)-1) \bigr|^s  \Bigr\}
\bigl| G_{\Gamma_1} (z;x + 2(N+n),y) \bigr|^s \\
&\leq C \bigl| G_{\Gamma_1} (z;x + 2(N+n),y) \bigr|^s
\end{eqnarray*}
where $\Gamma_0 = \ZZ$ and $\Gamma_1 = \{x + 2(N+n), x + 2(N+n)+1, \dots\}$.
In an appropriate large disorder regime, where the constant $C$ in \eref{e:K}
is smaller than one, exponential decay now follows by iteration, similarly as in Theorem
\ref{theorem:exp}.
\ack
It is a pleasure to thank Prof. G\"unter Stolz for
stimulating discussions. This work was partially supported by NSF
grant DMS--0907165 (AE) and by the Deutsche Forschungsgemeinschaft 
within the Emmy-Noether-Programme (MT and IV). 
The authors enjoyed the hospitality of the
Mathematisches Forschungsinstitut Oberwolfach  during the final
stage of this work.
\section*{References}
%
%
%

\end{document}